\RequirePackage{fix-cm}
\documentclass[smallcondensed]{svjour3}     % onecolumn (ditto)
\smartqed  % flush right qed marks, e.g. at end of proof
\usepackage{graphicx}
\usepackage{amssymb}
\usepackage{amsmath}

\usepackage{natbib}
\usepackage[font=small]{caption}
\usepackage{dcolumn}

\newtheorem{prop}{Proposition}%[section]
\newtheorem{cor}{Corollary}%[section]

\newcommand{\lp}{\left}
\newcommand{\rp}{\right}
\usepackage{epstopdf}
	\def\makeheadbox{\relax}
	\usepackage[margin=1.0in]{geometry}

\begin{document}

\title{Calculating CVaR and bPOE for Common Probability Distributions With Application to Portfolio Optimization and Density Estimation}

\titlerunning{Calculating CVaR and bPOE}        % if too long for running head

\author{Matthew Norton~\and~Valentyn Khokhlov~\and~Stan~Uryasev}


\institute{Matthew Norton \at
              Naval Postgraduate School, Operations Research Department \\
              \email{mnorton@nps.edu}           
           \and
           V. Khokhlov \at \email{vkhokhlov.embals2016@london.edu}
           \and
           S. Uryasev \at 
           University of Florida, Department of Industrial and Systems Engineering, Risk Management and Financial Engineering Laboratory \\
           \email{uryasev@ufl.edu}             
}

\date{}

\maketitle

\begin{abstract}
Conditional Value-at-Risk (CVaR) and Value-at-Risk (VaR), also called the superquantile and quantile, are frequently used to characterize the tails of probability distribution's and are popular measures of risk in applications where the distribution represents the magnitude of a potential loss. Buffered Probability of Exceedance (bPOE) is a recently introduced characterization of the tail which is the inverse of CVaR, much like the CDF is the inverse of the quantile. These quantities can prove very useful as the basis for a variety of risk-averse parametric engineering approaches. Their use, however, is often made difficult by the lack of well-known closed-form equations for calculating these quantities for commonly used probability distribution's. In this paper, we derive formulas for the superquantile and bPOE for a variety of common univariate probability distribution's. Besides providing a useful collection within a single reference, we use these formulas to incorporate the superquantile and bPOE into parametric procedures. In particular, we consider two: portfolio optimization and density estimation. First, when portfolio returns are assumed to follow particular distribution families, we show that finding the optimal portfolio via minimization of bPOE has advantages over superquantile minimization. We show that, given a fixed threshold, a single portfolio is the minimal bPOE portfolio for an entire class of distribution's simultaneously. Second, we apply our formulas to parametric density estimation and propose the method of superquantile's (MOS), a simple variation of the method of moment's (MM) where moment's are replaced by superquantile's at different confidence levels. With the freedom to select various combinations of confidence levels, MOS allows the user to focus the fitting procedure on different portions of the distribution, such as the tail when fitting heavy-tailed asymmetric data.  
\end{abstract}

\section{Introduction}
When faced with randomness and uncertainty, some of the most popular techniques for dealing with such randomness are parametric in nature. Given a real valued random variable $X$, analysis can be greatly simplified if one assumes that $X$ belongs to a specific parametric family of distribution's. For example, the Method of Moment's (MM) is one of the simplest and most widely used methods for parametric density estimation. These techniques, however, often require that certain characteristics of the distribution family be representable by a simple, ideally closed-form, expression. For example, traditional MM uses closed-form expression's for the moments of the parametric distribution family. Similarly, the Matching of Quantile's (MOQ) procedure (see e.g., \cite{sgouropoulos2015matching,karian1999fitting}) uses expression's for the quantile function. In portfolio optimization, the availability of simple expression's for the mean and variance of portfolio returns yields a tractable Markowitz portfolio optimization problem.\footnote{See Section 4 for specifics.} For a variety of problem's, application of a parametric method relies upon the availability of a closed-form expression for a specific characteristic of the parametric family of interest. 

Luckily, for a variety of distribution's, closed-from expression's are available for commonly utilized characteristic's. These include characteristic's such as the moment's, the quantile, and the CDF. Over the past two decades, however, new fundamental characteristic's like the superquantile have emerged from the field of quantitative risk management with important applications across engineering fields like financial, civil, and environmental engineering. (see e.g., \cite{Rockafellar, Uryasev, davis2016analysis}). Furthermore, closed-form expressions for these characteristics, for a large variety of common parametric distribution families, have not been widely disseminated.  While emerging from specific engineering applications, some of these characteristics are very general and can be viewed as fundamental aspects of a random variable just like the mean or quantile. Thus, utilization of these characteristics within parametric methods is a natural consideration. To facilitate their use, however, we must develop closed-form expressions.\footnote{When closed-form expressions are not available, we look to provide simple calculation methods that might still be utilized within parametric methods.} 

We focus on developing these expressions for the superquantile and Buffered Probability of Exceedance (bPOE) for a variety of distribution families.  Developments in financial risk theory over the last two decades have heavily emphasized measurement of tail risk. After \cite{Artzner} introduced the concept of a \textit{coherent} risk measure, \cite{Uryasev} introduced the superquantile, also called Conditional Value at Risk (CVaR) in the financial literature. This began to be considered a preferable characterization of tail risk compared to the quantile, or Value-at-Risk (VaR). While some closed-form expression are available to use the superquantile within parametric procedures, see e.g., \cite{Uryasev, Landsman, Andreev}, the variety of distribution's discussed within each of these sources is limited.

We illustrate that for a variety of common distribution's, straightforward techniques such as integration of the quantile function obtain a closed-form expression for the superquantile that is easy to use within subsequent parametric methods. We attempt to include a variety, providing superquantile formulas for the Exponential, Pareto/Generalized Pareto (GPD), Laplace, Normal, LogNormal, Logistic, LogLogistic, Generalized Student-t, Weibull, and Generalized Extreme Value (GEV) distribution's. These provide examples varying from the exponentially tailed (Exponential, Pareto/GPD, Laplace), to the symmetric (Normal, Laplace, Logistic, Student-t), to the asymmetric heavier tailed (Weibull, LogLogistic, GEV) distribution's. While some of these formulas may exist elsewhere, we hope that this paper serves as a good resource for practitioners in search of superquantile formulas. 

%CVaR and bPOE references
While the superquantile has risen in popularity over the past decade, a related characteristic called Buffered Probability of Exceedance (bPOE) has recently been introduced, first by \cite{Rockafellar} in the context of Buffered Failure Probability and then generalized by \cite{mafusalov2018buffered}. This concept has grown in popularity within the risk management community with application in finance, logistics, analysis of natural disasters, statistics, stochastic programming, and machine learning (\cite{shang2018cash, uryasev2014buffered, davis2016analysis, mafusalov2018estimation, norton2017soft, norton2016maximization}. Specifically, bPOE is the inverse of the superquantile in the same way that the CDF is the inverse of the quantile. However, much like the superquantile when compared against the quantile, bPOE has many mathematically advantageous properties over the traditionally used Probability of Exceedance (POE). Direct optimization often reduces to convex or linear programming, it can be calculated via a one dimensional convex optimization problem, and it provides a risk-averse probabilistic assessment of the risk of experiencing outcomes larger than some fixed upper threshold. Thus, the second aim of this paper is to provide closed-form expressions for bPOE and, when unable to do so, show that calculation of bPOE is still simple, reducing to a one-dimensional convex optimization problem or a one-dimensional root finding problem. For the parametric portfolio application, in particular, we will see that when closed-form bPOE is unavailable and the superquantile is available, finding the optimal bPOE portfolio is no more difficult, computationally, than finding the optimal superquantile (CVaR) portfolio.

Motivating us to derive closed-form expressions (or simple calculation formulas) for the superquantile and bPOE for common distribution's is the inclusion of these risk averse, tail measurements within parametric methods. In particular, we explore the use of the superquantile and bPOE within parametric portfolio optimization and density estimation. First, we consider parametric portfolio optimization, where returns are assumed to follow a specific distribution and, using these assumptions, a tractable portfolio optimization problem is formulated and solved. We begin by narrowing our choices of distribution to only those that both fit the pattern of portfolio returns and generate tractable portfolio optimization problems. Then, we consider two companion problems, solving for portfolio's that minimize the superquantile (CVaR) of the distribution of potential losses (i.e. the average of the worst-case $100(1-\alpha)\%$ scenarios) and portfolio's that minimize bPOE of the loss distribution (i.e. the \textit{buffered} probability that losses will exceed a fixed upper threshold $x$). In comparing these problems, we discover that bPOE optimization can often be highly preferable to superquantile (CVaR) optimization in the parametric context. Specifically, for fixed $\alpha$, the portfolio that minimizes the superquantile depends upon the distributional assumption (i.e., even if $\alpha$ is fixed, changing the assumed parametric distribution for returns will change the contents of the optimal portfolio).  However, for fixed threshold $x$, the portfolio that minimizes bPOE does not depend upon the distributional assumption (at least for the specific class of distribution's we consider, which includes the Logistic, Laplace, Normal, Student-t, and GEV).  In other words, no matter which of these distribution's we choose, we will always achieve the same optimal portfolio for fixed value of threshold $x$. Thus, bPOE-based portfolio optimization can provide additional consistency with respect to parameter choices, eliminating one source of additional variability for the decision maker. 

Finally, we consider parametric density estimation, proposing a variant of MM where moments are replaced by superquantile's. This can also be seen as a natural variation of the MOQ procedure where quantiles are replaced by superquantile's. Made possible by the closed-form superquantile expressions, we show that this framework allows one to flexibly perform density estimation, allowing the user to focus the fitting procedure on specific portions of the distribution. For example, we illustrate by fitting a Weibull with additional emphasis put onto estimating the right tail. Compared against traditional MM and maximum likelihood (ML), we see that we get a better fit for such asymmetric, heavy tailed situations. 

\subsection{Organization of Paper}
We first provide a brief introduction to superquantile's and bPOE in Section 1.2. In Section 2, we give formulas for both the superquantile and bPOE for the Exponential, Pareto, Generalized Pareto, and Laplace distribution's. Along the way, we highlight some simple relationships between POE, bPOE, the quantile, and the superquantile. In Section 3, we treat distribution's for which a closed-form superquantile formula exists, but where we are unable to derive a simple closed-form bPOE formula. In order of appearance, we consider the Normal, LogNormal, Logistic, Generalized Student-t, Weibull, LogLogistic, and Generalized Extreme Value Distribution. However, we point out for these cases that because a formula for the superquantile is known, bPOE can be solved for via a simple root finding problem. We also illustrate for some cases that the one-dimensional convex optimization formula for bPOE can also be used in these cases. In Section 4, we illustrate the use of these formulas in portfolio optimization and parametric distribution approximation.

\subsection{Background and Notation}

When working with optimization of tail probabilities, one frequently works with constraints or objectives involving \textit{probability of exceedance} (POE), $p_x (X)=P(X > x)$, or its associated quantile $q_{\alpha}(X)=\min \{x | P(X\leq x)\geq \alpha \}$, where $\alpha \in [0,1]$ is a probability level. The quantile is a popular measure of tail risk in financial engineering, but when included in optimization problems via constraints or objectives, is quite difficult to treat with continuous (linear or non-linear) optimization techniques.

A significant advancement was made in \cite{Uryasev,rockafellar2002conditional} in the development of a replacement called the superquantile or CVaR. The superquantile is a measure of uncertainty similar to the quantile, but with superior mathematical properties. Formally, the superquantile (CVaR) for a continuously distributed $X$ is defined as,
$$\bar{q}_\alpha (X) =E \lp[ X | X > q_\alpha (X) \rp]= \frac{1}{1-\alpha} \int_{q_\alpha (X)}^{\infty} xf(x)dx = \frac{1}{1-\alpha} \int_{\alpha}^{1} q_p (X) dp.$$
Similar to $q_\alpha (X)$, the superquantile can be used to assess the tail of the distribution. The superquantile, though, is far easier to handle in optimization contexts. It also has the important property that it considers the magnitude of events within the tail. Therefore, in situations where a distribution may have a heavy tail, the superquantile accounts for magnitudes of low-probability large-loss tail events while the quantile does not account for this information.

The notion of \textit{buffered probability} was originally introduced by \cite{Rockafellar} in the context of the design and optimization of structures as the Buffered Probability of Failure (bPOF). Working to extend this concept, bPOE was developed as the inverse of the superquantile by \cite{mafusalov2018buffered} in the same way that POE is the inverse of the quantile. Specifically, for continuously distributed $X$, bPOE at threshold $x$ is defined in the following way, where $\sup X$ denotes the essential supremum of random variable $X$ and threshold $x \in [ E[X] , \sup X]$.

$$\bar{p}_x (X)= \{ 1-\alpha | \bar{q}_\alpha (X) = x \}  \;.$$

In words, bPOE calculates one minus the probability level at which the superquantile, the tail expectation, equals the threshold $x$. Roughly speaking, bPOE calculates the proportion of worst-case outcomes which average to $x$. Figure~\ref{bPOE_illustrate} presents an illustration of bPOE for a Lognormal distributed random variable $X$. We note that there exist two slightly different variants of bPOE, called Upper and Lower bPOE which are identical for continuous random variables. For this paper, we utilize only continuous random variables. For the interested reader, details regarding the difference between Upper and Lower bPOE can be found in \cite{mafusalov2018buffered}.

\begin{figure}[t]
\centering
\includegraphics[width=2.52in,height=2.25in]{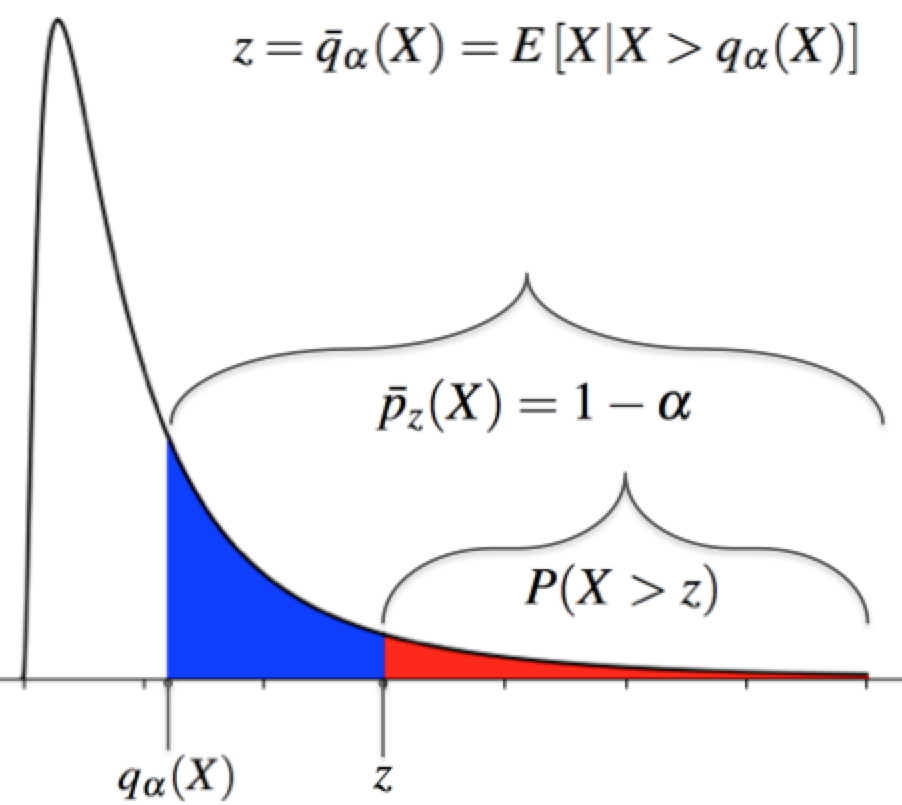} 
 \caption{Shown is the Probability Density Function (PDF) of $X\sim Lognormal(\sigma=1,\mu=0)$. Given threshold $z \in \mathbb{R}$, POE equals $P(X>z)$ the cumulative density in red. For the same threshold $z$, bPOE equals $\bar{p}_z(X)$ the combined cumulative density in red \textit{and} blue. By definition, the expectation of the worst-case $1-\alpha=\bar{p}_z(X)$ outcomes equals $z=\bar{q}_\alpha(X)$. These worst-case outcomes are those that are larger than the quantile $q_\alpha(X)$.}
 \label{bPOE_illustrate}
\end{figure}

Similar to the superquantile, bPOE is a more robust measure of tail risk, as it considers not only the probability that events/losses will exceed the threshold $x$, but also the magnitude of these potential events. Also, much like the superquantile, bPOE can be represented as the unique minimum of a one-dimensional convex optimization problem with the formulas given by \cite{norton2016maximization,mafusalov2018buffered} as follows, where $[\cdot]^+=\max\{\cdot,0\}$.
\begin{align*}
 \bar{p}_x (X) &= \min_{a \geq 0} E[ a(X-x) +1 ]^+ = \min_{\gamma <x} \frac{ E[X - \gamma]^+ } { x - \gamma} ,\\
 & \\
 \bar{q}_\alpha (X)&=\underset{\gamma}{\text{min }  } \gamma + \frac{ E[ X-\gamma ]^+} {1-\alpha} \; .
 \end{align*}
Note that these formulas are valid for general real valued random variables, not only continuously distributed random variables. It is also useful to note that the argmin of both the bPOE and superquantile optimization formulas gives the quantile. For the bPOE calculation formula, we have that the argmin is $\gamma^* = q_\alpha (X)$ where $\alpha = 1 - \bar{p}_x(X)$ and $a^* = \frac{1}{x-\gamma^*}$ for the other representation. For the superquantile calculation formula, we have that the argmin is $\gamma^* = q_\alpha(X)$ where $\alpha$ was the desired probability level for calculating the superquantile.

The bPOE concept is also closely related to the concept of a superdistribution function $\bar{F}(x)$, introduced by \cite{rockafellar2014random}. For the CDF, we have that POE equals $P(X>x) = 1 - F(x)$ and we have the inverse CDF given by $ F^{-1}(\alpha)=q_\alpha  (X).$ The superdistribution function $\bar{F}(x)$ is motivated by the inverse relation $ \bar{F}^{-1}(\alpha)=\bar{q}_\alpha  (X).$ Thus, bPOE equals $1-\bar{F}(x)$. The superdistribution function of a random variable $X$ can also be understood as the CDF of an auxiliary random variable $\bar{X}=\bar{q}_u (X)$ where $u \sim U(0,1)$ is a uniformly distributed random variable. In this case, $\bar{F}_X(x)=F_{\bar{X}}(x)$ where the subscript indicates that it is the distribution function associated with a particular random variable.

\section{Distributions With Closed Form Superquantile and bPOE}
In this section, we derive closed-form expressions for both the superquantile and bPOE for the Exponential, Pareto, Generalized Pareto, and Laplace distribution's. For these distribution's, we see that they exhibit a reproducing type of property, where the formula for POE is identical to bPOE up to a constant. The Laplace distribution presents an interesting case in which only the right tail exhibits this reproducing property. Along the way, for completeness, we also highlight relationships between the expressions for bPOE, POE, the superquantile, and the quantile.
%%%%%%%%%%%%%%%%%%%%%%%%%%%%%%%%%%%%%%%%%%%%%%%%%%%%%%%%%%%%%%%%%%%%%%%%%%%%%%%%%%%%%%%%%%%%%%%%%%%%%%%%%%%%%%%%%%%%%%%%%%%%
\subsection{Exponential}
For this section, we have exponential random variable $X\sim Exp(\lambda)$. Recall that the Exponential parameter has range $\lambda >0$ with $E[X]=\sigma(X)=\frac{1}{\lambda}$, and that the Exponential CDF, PDF, and quantile are given by, 
\[  F(x) = \begin{cases}
1-e^{-\lambda x} & x \ge 0, \\
0 & x < 0.
\end{cases}, \quad f(x)= \begin{cases}
\lambda e^{-\lambda x} & x \ge 0, \\
0 & x < 0.
\end{cases},
\quad q_\alpha(X) = \frac{-\ln (1 - \alpha ) }{\lambda}
\]
\begin{prop} \label{prop:exp} Let $X\sim Exp(\lambda)$. Then, 
\[. \bar{q}_\alpha (X) =  \frac{ - \ln(1-\alpha) + 1 } { \lambda} , \quad \bar{p}_x (X) = e^{1-\lambda x}.  \] 
\end{prop}
\begin{proof}
First, note that $q_\alpha(X) = \frac{-\ln (1 - \alpha ) }{\lambda}$ for exponential RV's with rate parameter $\lambda$. We then have,
\begin{align*}
\bar{q}_\alpha (X) &=  \frac{1}{1-\alpha} \int_{\alpha}^{1} q_p (X) dp \\
&=  \frac{-1}{\lambda(1-\alpha)} \int_{\alpha}^{1} \ln (1 - p )  dp \\
&=  \frac{-1}{\lambda(1-\alpha)} \int_{1-\alpha}^{0} -\ln (y )  dy =  \frac{-1}{\lambda(1-\alpha)} \int_{0}^{1 - \alpha} \ln (y )  dy
\end{align*}

Since $\int \ln(y)dy = y \ln(y) - y + C$, we have,
\begin{align*}
\bar{q}_\alpha (X) &=   \frac{-1}{\lambda(1-\alpha)} \int_{0}^{1 - \alpha} \ln (y )  dy \\
&=   \frac{-1}{\lambda(1-\alpha)}  \lp[  ( 1-\alpha) \ln(1-\alpha) - (1-\alpha) \rp] =   \frac{ - \ln(1-\alpha) + 1 } { \lambda} 
\end{align*}

We can then see that,
\begin{align*}
\bar{p}_x (X) &=  \{ 1 - \alpha | \bar{q}_\alpha (X) = x \} \\
&=  \{ 1 - \alpha |\frac{ - \ln(1-\alpha) + 1 } { \lambda} = x \} \\
&=  \{ 1 - \alpha |  \ln(1-\alpha)  =  1-\lambda x \} \\
&=  \{ 1 - \alpha |  e^{\ln(1-\alpha)}  =  e^{1-\lambda x} \} =  \{ 1 - \alpha |  1-\alpha =  e^{1-\lambda x} \} = e^{1-\lambda x} 
\end{align*}
\qed
\end{proof}
Next, we relate bPOE and POE as well as the superquantile and quantile.
\begin{cor}\label{cor:exp}
Let $X \sim Exp(\lambda)$, with mean $\mu = \frac{1}{\lambda}$. Then, $\bar{p}_x (X) = P(X > x -\mu)$ and  $\bar{q}_\alpha(X) =  q_\alpha(X) + \mu$.
\end{cor}
\begin{proof}
We know that $X$, being exponential, has CDF given by $P(X \geq x) = 1 - e^{-\lambda x}$. From Proposition 1, we know that 
\[
\bar{p}_x (X) = e^{(1 - \lambda x)} = e^{-\lambda(\frac{-1}{\lambda} + x)}\;.
\]

Then, since $\mu =\frac{1}{\lambda}$, it follows that $\bar{p}_x (X) = e^{-\lambda( x - \mu)} = 1 - P(X \leq x - \mu) = P(X > x - \mu)$. The equality for CVaR follows easily from Proposition 1 since  $q_\alpha(X) = \frac{-\ln (1 - \alpha ) }{\lambda}$.
\qed
\end{proof}

%%%%%%%%%%%%%%%%%%%%%%%%%%%%%%%%%%%%%%%%%%%%%%%%%%%%%%%%%%%%%%%%%%%%%%%%%%%%%%%%%%%%%%%%%%%%%%%%%%%%%%%%%%%%%%%%%%%%%%%%%%%%
\subsection{Pareto}
Assume $X \sim Pareto(a,x_m)$. Recall that Pareto parameters have range $a>0, x_m>0$ with $E[X]=\begin{cases} \infty ,& a \leq 1 ,\\ \frac{a x_m}{a-1} ,& a>1 \end{cases}$ and $\sigma^2(X)=\begin{cases} \infty ,& a \in(0,2] ,\\ \frac{a x_m^2}{(a-1)^2 (a-2)} ,& a>2 \end{cases}$, and that the Pareto CDF, PDF, and quantile are given by,
\[  F(x) = \begin{cases}
1-\lp( \frac{x_m}{x} \rp)^a & x \ge x_m, \\
0 & x < x_m.
\end{cases} \quad,\; f(x)= \begin{cases}
\frac{ax_m^a}{x^{a+1}} & x \ge x_m, \\
0 & x < x_m.
\end{cases}
\quad,\; q_\alpha(X) = \frac{x_m }{(1-\alpha)^\frac{1}{a} }
\]

\begin{prop} \label{prop:pareto} Assume $X \sim Pareto(a,x_m)$ with $a >1$. Then, for $\alpha \in [0,1]$ and $x \geq E[X]$,
\[ \bar{q}_\alpha (X) = \frac{ x_m a }{ (1-\alpha)^{\frac{1}{a}}  (a-1)}  \;,\quad \bar{p}_x(X) = \lp( \frac{x_m a}{x(a-1) }  \rp)^a \;.
\]
Note that if $a \in (0,1]$, then $E[X] = \infty$ implying that $\bar{q}_\alpha (X) = \infty$ and $\bar{p}_x(X) =1$ for all $\alpha \in [0,1]$ and $x \in \mathbb{R}^n$.
\end{prop}
\begin{proof}
First, note that the conditional distribution of a pareto, conditioned on the event that the random value is larger than some $\gamma$, is simply another pareto with parameters $a,\gamma$. This implies that $E[ X | X > \gamma ] = \frac{ a \gamma}{ a - 1}$ if $a\geq 1$ otherwise the expectation is $\infty$. Also, $1- F(\gamma) = \lp( \frac{x_m}{\gamma} \rp)^a$. Since, 
$$ E[ X  - \gamma ]^+ = (E[ X | X > \gamma ] - \gamma)(1- F(\gamma)) \;,$$
we will have that,
$$ E[ X  - \gamma ]^+ = (\frac{ a \gamma}{ a - 1}- \gamma) \lp( \frac{x_m}{\gamma} \rp)^a \;.$$
This gives us bPOE formula,
\begin{align*}
 \bar{p}_x(X) &= \min_{x_m\leq\gamma<x} \frac{  \lp( \frac{ a \gamma}{ a - 1}- \gamma \rp) x_m^a }{ \gamma^a (x-\gamma) } \\
 &= \min_{x_m\leq\gamma<x} \frac{  \lp( \frac{ a }{ a - 1}- 1 \rp) x_m^a }{ \gamma^{a-1} (x-\gamma) } = \lp( \max_{x_m\leq\gamma<x} \frac{\gamma^{a-1} (x-\gamma)(a-1) }{x_m^a} \rp)^{-1}
\end{align*}
Since $a>1$, the maximization objective is concave over the range $\gamma \in (0,\infty)$ which contains the range $(x_m,x)$, so we just need to take the gradient of function $g(\gamma)=\frac{\gamma^{a-1} (x-\gamma)(a-1) }{x_m^a}$ and set it to zero to find the optimal $\gamma$ as follows:
\begin{align*}
 \frac{\partial g}{\partial \gamma} = \frac{ x(a-1)^2 \gamma^{a-2} - (a-1)a\gamma^{a-1} }{x_m^a} =0 &\implies x(a-1)^2 \gamma^{a-2} = (a-1)a\gamma^{a-1} \\ 
& \implies \frac{x(a-1)}{a} = \gamma \\
  \end{align*}
  Plugging this value of $\gamma$ into the objective of our bPOE formula yields,
\begin{align*}
 \bar{p}_x(X) &=\frac{  \lp( \frac{ \frac{ax(a-1)}{a}  }{a-1} -    \frac{x(a-1)}{a}     \rp) x_m^a  } {  \lp(\frac{x(a-1)}{a}\rp)^a \lp( x -  \frac{x(a-1)}{a} \rp) } \\
&  = \lp( \frac{x_m a}{x(a-1) }  \rp)^a
\end{align*}
CVaR is then equal to the value of $x$ which solves the equation $1-\alpha = \bar{p}_x(X)$ or,
\[ 1-\alpha = \lp( \frac{x_m a}{x(a-1) }  \rp)^a \;,
\]
which has solution,
\[ \bar{q}_\alpha (X) = \frac{ x_m a }{ (1-\alpha)^{\frac{1}{a}}  (a-1)} \;. 
\]
\qed
\end{proof}
\begin{cor} \label{cor:pareto}
Relating bPOE and POE, as well as the quantile and superquantile, we can say that,
\[ \bar{p}_x(X) = P(X> \frac{x(a-1)}{a}) = P(X > x) \lp( \frac{a}{a-1} \rp)^a \\ 
\text{ and    } \bar{q}_\alpha (X) = q_\alpha (X) \frac{ a }{ a-1}\;.
\]
\end{cor}
\begin{proof}
Follows from Proposition 1 and the known formulas for POE and the quantile.
\qed
\end{proof}
%%%%%%%%%%%%%%%%%%%%%%%%%%%%%%%%%%%%%%%%%%%%%%%%%%%%%%%%%%%%%%%%%%%%%%%%%%%%%%%%%%%%%%%%%%%%%%%%%%%%%%%%%%%%%%%%%%%%%%%%%%%%
\subsection{Generalized Pareto Distribution (GPD)}
Assume $X \sim GPD(\mu,s,\xi)$. Recall that GPD parameters have range $\mu \in \mathbb{R}, s>0, \xi \in \mathbb{R}$ with $E[X]= \mu + \frac{s}{1-\xi}$ if $\xi <1$ and $\sigma^2(X)= \frac{s^2}{(1-\xi)^2(1-2\xi)}$ if $\xi <.5$, and that the GPD CDF and PDF are given by,
\[  F(x) = {\begin{cases}1-\left(1+{\frac  {\xi (x-\mu )}{s }}\right)^{{-1/\xi }}&{\text{for }}\xi \neq 0,\\1-\exp \left(-{\frac  {x-\mu }{s }}\right)&{\text{for }}\xi =0.\end{cases}} \quad,\; f(x)={\frac  {1}{s }}\left(1+{\frac  {\xi (x-\mu )}{s }}\right)^{{\left(-{\frac  {1}{\xi }}-1\right)}}\;.
\]
for $x \geq \mu$ when $\xi \geq 0$ and $\mu \leq x \leq \mu - \frac{s}{\xi}$ when $\xi <0$. Furthermore, the quantiles are given by,
\[
q_\alpha(X) = \begin{cases}\mu + \frac{s\lp(  (1-\alpha)^{-\xi} - 1 \rp)}{\xi}&{\text{for }}\xi \neq 0,\\\mu - s \ln(1-\alpha) &{\text{for }}\xi =0.\end{cases} \;
\]

\begin{prop} \label{prop:gen_pareto} Assume $X \sim GPD(\mu,s,\xi)$ with $-1<\xi <1$. Then,
\[ \bar{q}_\alpha (X) = \begin{cases}\mu + s\lp[ \frac{(1-\alpha)^{-\xi} }{1-\xi}   + \frac{(1-\alpha)^{-\xi} -1 }{\xi}  \rp]&{\text{for }}\xi \neq 0,\\\mu + s[1- \ln(1-\alpha) ]&{\text{for }}\xi =0.\end{cases} \;,\] \[ \bar{p}_x(X) = \begin{cases}\frac{  \lp(1+\frac{\xi(x-\mu)}{s}\rp)^{- \frac{1}{\xi} }  }{(1-\xi)^{ \frac{1}{\xi} } } &{\text{for }}\xi \neq 0,\\\ e^{ 1 - \lp( \frac{x-\mu}{s} \rp) }&{\text{for }}\xi =0.\end{cases} \;.
\]

\end{prop}
\begin{proof}
For these results, we rely on the fact that if $X \sim GPD(\mu,s,\xi)$, then $X-\gamma | X> \gamma \sim GPD(0,s+\xi(\gamma-\mu),\xi)$, meaning that the excess distribution of a GPD random variable is also GPD. Now, note also that if $\xi<1$, then $E[X]=\mu + \frac{s}{1-\xi}$. This gives us,
\[  E[X-\gamma | X>\gamma] = E[ GPD(0,s+\xi(\gamma-\mu),\xi)] = \frac{s+\xi(\gamma - \mu)}{1-\xi}
\]
which further implies that,
\begin{align*}
\bar{q}_\alpha(X) &= E[X-q_\alpha(X) | X>q_\alpha(X)] + q_\alpha(X) \\
&= \frac{s + \xi( q_\alpha(X) - \mu )}{1-\xi} +q_\alpha(X) \;.
\end{align*}
Plugging in the values of the quantile functions yields the final formulas. Using the formulas we just found for $\bar{q}_\alpha(X)$, it is an elementary exercise to solve for $\bar{p}_x(X)$ which equals $1-\alpha$ such that $\alpha$ solves the equation $x=\bar{q}_\alpha(X)$.
\qed
\end{proof}%%%%%%%%%%%%%%%%%%%%%%%%%%%%%%%%%%%%%%%%%%%%%%%%%%%%%%%%%%%%%%%%%%%%%%%%%%%%%%%%%%%%%%%%%%%%%%%%%%%%%%%%%%%%%%%%%%%%%%%%
\subsection{Laplace}
Assume $X \sim Laplace(\mu,b)$. Recall that Laplace parameters have range $\mu \in \mathbb{R}$, $b>0$ with $E[X]=\mu$ and $\sigma^2(X)=2b^2$, and that the Laplace CDF, PDF, and quantile function are given by,
\[  F(x) = \begin{cases}
1-\frac{1}{2}e^{-\frac{x-\mu}{b}} & x \ge \mu, \\
\frac{1}{2}e^{\frac{x-\mu}{b}} & x < \mu.
\end{cases} \quad,\; f(x)= \frac{1}{2b} e^{\frac{-|x-\mu|}{b}}
\quad,\; \]
\[q_\alpha(X) = \mu - b\,sign(\alpha-0.5)\,\ln(1 - 2|\alpha-0.5|)\;.\]

% For the Laplace distribution, we can find a closed-form representation of the superquantile. For bPOE, if we have threshold $x\geq \mu +b$, then bPOE has a closed-form solution. However, if $x< \mu +b$ we need a special function to calculate bPOE. This, though, is not a terrible situation. Note that $x\geq \mu +b$ if and only if $\bar{p}_x(X) \leq .5$, meaning that we are only considering the upper half of the tail of the distribution. Characterizing the upper tail is the primary use of bPOE and the superquantile, therefore this case is the most important and has a closed-form solution. 
\begin{prop} \label{prop:laplace}  If $X \sim Laplace(\mu,b)$, then
\[ \bar{q}_\alpha (X) = \begin{cases}
\mu +b \lp(\frac{\alpha}{1-\alpha} \rp)(1 - \ln (2 \alpha)) & \alpha < .5, \\
\mu + b\lp(1 - \ln \lp(2(1- \alpha)\rp) \rp)  & \alpha \geq .5.
\end{cases}  \;,
\] \[
\bar{p}_x (X) = \begin{cases}
 \frac{1}{2}e^{1-\lp(\frac{x-\mu}{b} \rp)}& x \geq \mu +b, \\
1 + \frac{z}{\mathcal{W}( -2e^{-z-1}z)} & x < \mu + b.
\end{cases}
\]
where $z=\frac{x-\mu}{b}$ and $\mathcal{W}$ is the Lambert-$\mathcal{W}$ function.\footnote{Also called the product logarithm or omega function.}
\end{prop}
\begin{proof}
To get the superquantile, we begin with the integral representation:
\begin{align*}
\bar{q}_\alpha (X) &=  \frac{1}{1-\alpha} \int_{\alpha}^{1} q_p (X) dp \\
&=  \frac{1}{1-\alpha} \int_{\alpha}^{1}  \mu - b\,sign(p-0.5)\,\ln(1 - 2|p-0.5|) dp \\
&=   \mu -  \frac{b}{1-\alpha} \int_{\alpha}^{1} sign(p-0.5)\,\ln(1 - 2|p-0.5|) dp \\
&=   \mu -  \frac{b}{1-\alpha} \lp( \int_{\min\{\alpha,.5\}}^{.5} -\ln(2p) dp  + \int_{\max\{\alpha,.5\}}^{1} \ln(2(1-p)) dp   \rp) \;. 
\end{align*}
To evaluate the integral, we utilize simple substitution as well as the identity $\int \ln(y)dy = y \ln(y) - y + C$. After simplifying, we see that with $\alpha < .5$ the integral evaluates to,
\[ \bar{q}_\alpha (X)= \mu +b \lp(\frac{\alpha}{1-\alpha} \rp)(1 - \ln (2 \alpha)) \;.
\]
Similarly, we find that with $\alpha \geq .5$ the integral evaluates to,
\[ \bar{q}_\alpha (X) =\mu + b\lp(1 - \ln \lp(2(1- \alpha)\rp) \rp)\;.
\]
For bPOE, first assume that threshold $x \geq \mu +b$. Using our formula for CVaR, we see that $\bar{q}_{.5} (X) = \mu + b$. Thus, $x \geq \mu +b$ implies that $1-\bar{p}_x(X) \geq .5$ implying that,
\begin{align*}
\bar{p}_x (X) &=  \{ 1 - \alpha | \bar{q}_\alpha (X) = x , \alpha \geq .5\} \\
&=  \{ 1 - \alpha |\mu + b\lp(1 - \ln \lp(2(1- \alpha)\rp) \rp) = x \} \\
&= \frac{1}{2}e^{1-\lp(\frac{x-\mu}{b} \rp)} \;.
\end{align*}
Assume contrarily that $x< \mu + b$. Since $\bar{q}_{.5} (X) = \mu + b$, we have that $1-\bar{p}_x(X) < .5$ which implies that,
\begin{align*}
\bar{p}_x (X) &=  \{ 1 - \alpha | \bar{q}_\alpha (X) = x , \alpha < .5\} \\
&=  \{ 1 - \alpha |\mu +b \lp(\frac{\alpha}{1-\alpha} \rp)(1 - \ln (2 \alpha))= x \} \;.
\end{align*}
Letting $z=\frac{x-\mu}{b}$, we must now find $\alpha$ which solves the equation $\lp(\frac{\alpha}{1-\alpha} \rp)(1 - \ln (2 \alpha))= z$. We do so, via the following:
\begin{align*}
\lp(\frac{\alpha}{1-\alpha} \rp)(1 - \ln (2 \alpha))= z &\implies \frac{-z}{\alpha} = \frac{( \ln (2 \alpha) - 1)}{1-\alpha} \\
&\implies  e^{\frac{-z}{\alpha}} = e^{\frac{(\ln (2 \alpha) - 1 )}{1-\alpha}} = \lp( \frac{2\alpha}{e} \rp)^\frac{1}{1-\alpha} \\
&\implies e^{\frac{-z(1-\alpha)}{\alpha}} = \lp( \frac{2\alpha}{e} \rp)\\
& \implies  \frac{-z}{\alpha} e^{-z(\frac{1}{\alpha} -1)}  = -2ze^{-1} \\
&\implies \frac{-z}{\alpha} e^{\frac{-z}{\alpha} } = -2ze^{-z-1} \\
&\implies \frac{-z}{\alpha} = \mathcal{W}(-2ze^{-z-1}) \;.
\end{align*}
where the final step follows from the definition of the Lambert-$\mathcal{W}$ function which is given by the relation $xe^x=y \iff  \mathcal{W}(y)=x$.
Thus, $\frac{-z}{\alpha} = \mathcal{W}(-2ze^{-z-1}) \implies \bar{p}_x(X) = 1-\alpha =1 + \frac{z}{\mathcal{W}( -2e^{-z-1}z)} $.
\qed
\end{proof}

%%%%%%%%%%%%%%%%%%%%%%%%%%%%%%%%%%%%%%%%%%%%%%%%%%%%%%%%%%%%%%%%%%%%%%%%%%%%%%%%%%%%%%%%%%%%%%%%%%%%%%%%%%%%%%%%%%%%%%%%

\section{Distributions With Closed Form Superquantile}
In this section, we derive closed-form expressions for the superquantile of the Normal, LogNormal, Logistic, Student-t, Weibull, LogLogistic, and GEV distribution's. The Normal, Logistic, and Student-t provide us with examples of symmetric distribution's with varying tail heaviness. The LogNormal, Weibull, LogLogistic, and GEV provide us with examples of asymmetric distribution's that have heavy right tails. In particular, we will utilize the Weibull formula for density estimation in Section 5. 

For these distribution's, we are not able to reduce calculation of bPOE to closed-form. However, we highlight for the case of the Normal and Logistic that bPOE can be calculated by solving a one-dimensional convex optimization problem or one-dimensional root finding problem. In general, we note that for continuous $X$, bPOE at $x$ equals $1-\alpha$ where $\alpha$ solves $\bar{q}_\alpha (X) =x$. Thus, if the superquantile is known in closed-form, this reduces to a simple one-dimensional root finding problem in $\alpha$.

%%%%%%%%%%%%%%%%%%%%%%%%%%%%%%%%%%%%%%%%%%%%%%%%%%%%%%%%%%%%%%%%%%%%%%%%%%%%%%%%%%%%%%%%%%%%%%%%%%%%%%%%%%%%%%%%%%%%%%%%%%%%

\subsection{Normal}
Let $X \sim \mathcal{N}(0,1)$ be a standard normal random variable. Recall that 
\[  F(x) =\frac{1}{2} \lp[1 + \text{erf}\lp(\frac{x}{\sqrt{2}} \rp) \rp] ,\quad f(x)= \frac{1}{\sqrt{2\pi}} e^{\frac{-x^2}{2} }, \quad q_\alpha(X)=\sqrt{2} \text{erf}^{-1}(2\alpha -1) \;,\]
where $\text{erf}(\cdot)$ is the commonly known error function with $\text{erf}^{-1}(\cdot)$ denoting its inverse. 

We show that the superquantile can be calculated by utilizing the quantile function and PDF, which is a well known result (see e.g., \cite{Uryasev}). We also show that bPOE can be calculated in two ways: by solving a simple root finding problem involving only the PDF and CDF or by solving a convex optimization problem with gradients calculated via the commonly used error function. Some results are presented only for the Standard Normal $\mathcal{N}(0,1)$, but can easily be applied to the non-standard case $\mathcal{N}(\mu,\sigma)$ with appropriate shifting and scaling. 
\begin{prop}  \label{prop:normal_cvar} 
If $X \sim \mathcal{N}(\mu,\sigma)$, then 
\[
\bar{q}_\alpha (X) = \mu + \sigma \frac{ f( q_\alpha(\frac{X-\mu}{\sigma} ) )}{1-\alpha} \;.
\]
\end{prop}
\begin{proof}
It is well known that if $X \sim \mathcal{N}(0,1)$, then the conditional expectation is given by the inverse Mills Ratio, $E[ X | X > \gamma] =  \frac{f(\gamma)}{1 - F(\gamma) }$. It follows then that $\bar{q}_\alpha (X) = E[ X | X > q_\alpha(X) ] = \frac{f(q_\alpha(X))}{1 - F(q_\alpha(X)) }= \frac{f(q_\alpha(X))}{1 - \alpha}$.
\qed
\end{proof}
%\begin{prop}
%If $X \sim \mathcal{N}(\mu,\sigma)$, then 
%\[
%\bar{p}_x (X) =e^{-\frac{(x-\mu)^2}{2\sigma}}\;.
%\]
%\end{prop}
%\begin{proof}
%Let $\bar{f}(x)$ denote the super-pdf of $X$ at $x$ and $\bar{F}(x)$ denote the super-CDF where $\bar{F}'(x)=\bar{f}(x)$ and $\bar{p}_x(X) = 1-\bar{F}(x)$. From Proposition~\ref{prop:normal_cvar}, we know that the function $\bar{F}$ is the solution to the first order linear ODE,
%\[
%x=\mu + \sigma \frac{ \bar{F}'}{1-\bar{F} } \;.
%\]
%Solving this ODE yields the solution $\bar{F}(x)=1- e^\frac{(x-\mu)^2}{2\sigma}$. 
%\end{proof}
\begin{prop}  \label{prop:normal_bPOE_min} 
If $X \sim \mathcal{N}(0,1)$, then 
$$ \bar{p}_x (X) = \min_{\gamma < x} \frac{ f(\gamma) - \gamma(1-F(\gamma) ) } { x- \gamma} \;.$$
Furthermore, if $\gamma \in argmin$, then $\gamma$ equals the quantile of $X$ at probability level $1 -  \bar{p}_x (X) $. 
\end{prop}
\begin{proof}
Note that for a standard normal random variable, the tail expectation beyond any threshold $\gamma$ is given by the inverse Mills Ratio,
$$ E[ X | X > \gamma ] = \frac{f(\gamma)}{1 - F(\gamma) } \;. $$
Note also that for any threshold $\gamma$ and any random variable we have,
$$ E[ X  - \gamma ]^+ = (E[ X | X > \gamma ] - \gamma)(1- F(\gamma)) \;.$$
Using the mills ratio gives us,
$$ E[ X  - \gamma ]^+ = (\frac{f(\gamma)}{1 - F(\gamma) } - \gamma)(1- F(\gamma)) = f(\gamma) - \gamma (1 - F(\gamma)) \;.$$
Plugging this result into the minimization formula for bPOE yields the final formula. 
\qed
\end{proof}
\begin{prop}\label{prop:normal_bPOE_root}
Let $X \sim \mathcal{N}(0,1)$ with $x \in \mathbb{R}$ given. If $\gamma$ is the solution to the equation 
$$ \frac{f(\gamma)}{1 - F(\gamma)} = x \; $$
then $\bar{p}_x (X) = \frac{ f(\gamma) - \gamma(1-F(\gamma) ) } { x- \gamma}$. Additionally, we will have that $q_\alpha (X) = \gamma$ and $\bar{q}_\alpha (X) =x$ at probability level $\alpha=1- \bar{p}_x (X)$.
\end{prop}
\begin{proof}
This follows from the fact that $\bar{q}_\alpha (X) = E[ X | X > q_\alpha(X) ] = \frac{f(q_\alpha(X))}{1 - F(q_\alpha(X)) }$ and the optimization formula of bPOE given in the previous proposition for normally distributed variables.
\qed
\end{proof}

The following proposition provides the gradient calculation for solving the bPOE minimization problem. 
\begin{prop} \label{prop:normal_bPOE_min_integral}
For $X \sim \mathcal{N}(0,1)$, we have that the bPOE minimization formula has the following integral representation,
\begin{align*}
 \bar{p}_x (X) &= \min_{\gamma < x} \frac{ f(\gamma) - \gamma(1-F(\gamma) ) } { x- \gamma} \\
 &= \min_{\gamma < x} \frac{1}{\sqrt{2 \pi}}  \int_0^\infty \frac{ u e^{ \frac{-(\gamma + u )^2}{2} } }{x-\gamma} du \\
&= \min_{\gamma < x} \frac{  e^{ \frac{-\gamma^2}{2} }- \gamma \sqrt{ \frac{\pi}{2} }  \text{erfc}(\frac{\gamma}{\sqrt{2}}) }   { \sqrt{2\pi} (x - \gamma)  }
 \end{align*}
Furthermore, the function $g(u,\gamma ; x) = \frac{1}{\sqrt{2 \pi}}  \int_0^\infty \frac{ u e^{ \frac{-(\gamma + u )^2}{2} } }{x-\gamma} du$ is convex w.r.t. $\gamma$ over the range $\gamma \in (-\infty, x)$. Additionally, $g$ has gradient given by,
\begin{align*} \frac{\partial g}{\partial \gamma} &=\frac{1}{\sqrt{2 \pi}}  \int_0^\infty \frac{\partial }{\partial \gamma}\lp( \frac{ u e^{ \frac{-(\gamma + u )^2}{2} } }{x-\gamma}  \rp) du \\
&= \frac{  e^{ \frac{-\gamma^2}{2} }- x \sqrt{ \frac{\pi}{2} }  \text{erfc}(\frac{\gamma}{\sqrt{2}}) }   { \sqrt{2\pi} (x - \gamma)^2  }
\end{align*}
where $\text{erfc}(\cdot)$ denotes the complementary error function.
\end{prop}
\begin{proof}
To derive the integral representation, simply plug in the formula for $E[ X  - \gamma ]^+$, then utilize the definition of the PDF and CDF. The gradient calculation is a standard calculus exercise.
\qed
\end{proof}

%%%%%%%%%%%%%%%%%%%%%%%%%%%%%%%%%%%%%%%%%%%%%%%%%%%%%%%%%%%%%%%%%%%%%%%%
\subsection{LogNormal}
Assume $X \sim LogNormal(\mu,s)$. Recall that LogNormal parameters have range $\mu \in \mathbb{R}$, $s>0$, with $E[X]=e^{\mu+\frac{s^2}{2} }$ and $\sigma^2(X)=(e^{s^2} -1) e^{2\mu+s^2} $ and that the LogNormal CDF, PDF, and quantile function are given by,
\[  F(x) =\frac{1}{2} \lp[1 + \text{erf}\lp(\frac{\ln x - \mu }{s \sqrt{2}} \rp) \rp] ,\quad f(x)= \frac{1}{xs\sqrt{2\pi}} e^{\frac{-(\ln x - \mu)^2}{2s^2} } ,\]
\[\quad q_\alpha(X)=e^{\mu + s\sqrt{2} \text{erf}^{-1}(2\alpha -1)} \;.\]

\begin{prop}  \label{prop:logistic_cvar} 
If $X \sim LogNormal(\mu,s)$, then
\[ \bar{q}_\alpha (X)  =  \frac{1}{2} e^{\mu + \frac{s^2}{2}} \frac{ \lp[ 1 + \text{erf}\lp( \frac{s}{\sqrt{2}} - \text{erf}^{-1}(2\alpha -1) \rp)  \rp] } { 1-\alpha }.\]
\end{prop}
\begin{proof} We simply evaluate the integral of the quantile function as follows.
\begin{align*}
\bar{q}_\alpha (X) &=  \frac{1}{1-\alpha} \int_{\alpha}^{1} q_p (X) dp \\
&= \frac{1}{1-\alpha}\int_{\alpha}^{1} e^{\mu + s\sqrt{2} \text{erf}^{-1}(2p -1)}  dp \\
&= \frac{e^{\mu }}{1-\alpha}\int_{\alpha}^{1} e^{ s\sqrt{2} \text{erf}^{-1}(2p -1)}  dp \\
&= \frac{e^{\mu }}{1-\alpha}  \lp[ -\frac{1}{2} e^{ \frac{s^2}{2}}  \lp( 1 + \text{erf}\lp( \frac{s}{\sqrt{2}} - \text{erf}^{-1}(2p -1) \rp)  \rp) \rp]_{p=\alpha}^1 \\
&= \frac{e^{\mu }}{1-\alpha}  \lp[ \frac{1}{2} e^{ \frac{s^2}{2}} + \frac{1}{2} e^{ \frac{s^2}{2}}  \lp( 1 + \text{erf}\lp( \frac{s}{\sqrt{2}} - \text{erf}^{-1}(2p -1) \rp)  \rp) \rp] \\
&= \frac{1}{2} e^{\mu + \frac{s^2}{2}} \frac{ \lp[ 1 + \text{erf}\lp( \frac{s}{\sqrt{2}} - \text{erf}^{-1}(2\alpha -1) \rp)  \rp] } { 1-\alpha }.
\end{align*}
\qed
\end{proof}
%%%%%%%%%%%%%%%%%%%%%%%%%%%%%%%%%%%%%%%%%%%%%%%%%%%%%%%%%%%%%%%%%%%%%%%%
\subsection{Logistic}
Assume $X \sim Logistic(\mu,s)$. Recall that Logistic parameters have range $\mu \in \mathbb{R}$, $s>0$, with $E[X]=\mu$ and $\sigma^2(X)=\frac{s^2 \pi^2}{3}$ and that the Logistic CDF, PDF, and quantile function are given by,
\[  F(x) ={\frac  {1}{1+e^{{-{\frac  {x-\mu }{s}}}}}} \;,\quad f(x)={\frac  {e^{{-{\frac  {x-\mu }{s}}}}}{s\left(1+e^{{-{\frac  {x-\mu }{s}}}}\right)^{2}}}  \;,\quad q_\alpha (X) = \mu + s \ln\lp( \frac{\alpha}{1-\alpha} \rp)  \;.\]

Here, we derive a closed-form expression for the superquantile for the logistic distribution and derive a simple root finding problem for calculating bPOE. We also find that these quantities have a correspondence with the binary entropy function. 
\begin{prop}  \label{prop:logistic_cvar} 
If $X \sim Logistic(\mu,s)$, then
$$ \bar{q}_\alpha (X) = \mu + \frac{sH(\alpha)}{1-\alpha} $$
where $H(\alpha)$ is the binary entropy function $H(\alpha) = -\alpha \ln(\alpha) - (1-\alpha) \ln (1-\alpha)$. Furthermore, for any $x \geq \mu$, if $\alpha$ solves the equation,
$$  \frac{H(\alpha)}{1-\alpha} = \frac{x - \mu }{s} \;,$$
then $\bar{p}_x (X) = 1- \alpha$. Additionally, $\bar{p}_x (X) = 1- \alpha$ if $\alpha$ is the solution to the transformed system,
$$ (1-\alpha)\alpha^{\frac{\alpha}{1-\alpha}} = e^{ -\lp(  \frac{x-\mu}{s} \rp) }  \;.$$
Note that both functions $\frac{H(\alpha)}{1-\alpha}$ and $ (1-\alpha)\alpha^{\frac{\alpha}{1-\alpha}}$ are one-dimensional, convex, and monotonic over the range $\alpha \in [0,1]$, and thus unique solutions exist and can easily be found via root finding methods. 
\end{prop}
\begin{proof}
To obtain the superquantile, we have
\begin{align*}
\bar{q}_\alpha (X) &=  \frac{1}{1-\alpha} \int_{\alpha}^{1} q_p (X) dp \\
&=  \frac{1}{1-\alpha} \int_{\alpha}^{1} \mu + s \ln\lp( \frac{\alpha}{1-\alpha} \rp) dp \\
&=   \mu +  \frac{s}{1-\alpha} \int_{\alpha}^{1} \ln(p) - \ln(1-p) dp \\
&= \mu +  \frac{s}{1-\alpha}\lp( \int_{\alpha}^{1} \ln(p) dp  +  \int_{\alpha}^{1}  - \ln(1-p) dp \rp)
\end{align*}
Utilizing simple substitution as well as the identity $\int \ln(y)dy = y \ln(y) - y + C$, we get
\begin{align*}
\bar{q}_\alpha (X) &= \mu + \frac{s}{1-\alpha} \lp( -1 - \alpha \ln \alpha + \alpha - (1-\alpha) \ln(1-\alpha) + (1-\alpha) \rp) \\
&=  \mu + \frac{s}{1-\alpha} \lp( - \alpha \ln \alpha  - (1-\alpha) \ln(1-\alpha)  \rp) \\
&=  \mu + \frac{s}{1-\alpha} H( \alpha) \;.
\end{align*}
To get bPOE, we simply follow the bPOE definition, needing to find $\alpha$ which solves $ \mu + \frac{s}{1-\alpha} H( \alpha)=x$. The transformed system arises from combining logarithms within the superquantile formula and applying exponential transformations.
\qed
\end{proof}
We can also utilize the minimization formula to calculate bPOE. Calculating bPOE in this way has the added benefit of simultaneously calculating the quantile $q_{1-\bar{p}_x(X)} (X)$.
\begin{prop}  \label{prop:logistic_bPOE} 
If $X \sim Logistic(\mu,s)$, then
\[ \bar{p}_x(X) = \min_{\gamma <x} \frac{s\ln( 1+e^{-(\frac{\gamma-\mu}{s}) } ) }{x-\gamma} \;,
\]
which is a convex optimization problem over the range $\gamma \in (-\infty,x)$. Furthermore, the minimum occurs at $\gamma$ such that,
\[ \frac{s\ln( 1+e^{-(\frac{\gamma-\mu}{s}) } ) }{x-\gamma}  = 1-F(\gamma)\;.
\]
\end{prop}
\begin{proof}
This follows from the fact that $E[X-\gamma]^+ = \int_{\gamma}^{\infty} (1-F(t))dt$. Evaluating this integral for $X \sim Logistic(\mu,s)$ yields, $E[X-\gamma]^+= s\ln( 1+e^{-(\frac{\gamma-\mu}{s}) } )$ which can then be plugged into the minimization formula for bPOE. The second part of the proposition follows from the fact that the gradient of the objective function w.r.t. $\gamma$ is given by,
\[  \frac{s\ln( 1+e^{-(\frac{\gamma-\mu}{s}) } )}{ (x-\gamma)^2 } - \frac{  e^{-(\frac{\gamma-\mu}{s} )} } { (x-\gamma)\lp( 1+e^{-(\frac{\gamma-\mu}{s} )} \rp) }  \;.
\]
Setting this gradient to zero and simplifying yields the stated optimality condition.
\qed
\end{proof}

%\begin{prop}  \label{prop:logistic_approx} 
%If $X \sim Logistic(\mu,s)$, then a very tight approximation for bPOE is given by,
%\[ \bar{p}_x(X) \approx \max\lp\{ .5\lp( tanh(\frac{z}{2}) -( sech(.85z) - 1) \rp), \frac{1}{1+e^{-z}} -  \frac{1.7}{1+e^{z}}  \rp\}\;,
%\]
%where $z=\frac{x-\mu}{s}$, $tanh(\cdot)$ is the hyperbolic tangent function, and $sech(\cdot)$ is the hyperbolic secant function.
%\end{prop}
%\begin{proof}
%From propositions before, we know that bPOE can be found by solving for the $\alpha$ which solves the equation $\frac{H(\alpha)}{1-\alpha} = \frac{x - \mu }{s}$. Thus, we need to find the inverse of the function $\frac{H(\alpha)}{1-\alpha} $. This approximation was found by visual inspection by plotting $\frac{H(y_1)}{1-y_1}=y_2$ on a 2d graph with axis $(y_1,y_2)$ and attempting to find the inverse function $g(y_2)=y_1$ which mimicked the graph of $\frac{H(y_1)}{1-y_1}=y_2$. The stated approximation was found to be very tight, particularly for large values of $x$ yielding small bPOE.
%\end{proof}
%%%%%%%%%%%%%%%%%%%%%%%%%%%%%%%%%%%%%%%%%%%%%%%%%%%%%%%%%%%%%%%%%%%%%%%%%%%%%%%%%%%%%%%%%%%%%%%%%%%%%%%%%%%%%%%%%%%%%%%%%%%%

%%%%%%%%%%%%%%%%%%%%%%%%%%%%%%%%%%%%%%%%%%%%%%%%%%%%
\subsection{Student$-t$}
Assume $X \sim Student$-t$(\nu,s,\mu)$. Recall that Student$-t$ parameters have range $\nu >0$, $s>0$, $\mu>0$ with $E[X]=\mu$ and $\sigma^2(X) = \frac{s^2 \nu}{\nu-2}$, and that the Student$-t$ CDF and PDF are given by,
\[   F(x) = 1- \frac{1}{2} \mathcal{I}_{\nu(x)} \left(\frac{\nu}{2} , \frac{1}{2}\right) \;,\quad f(x)= \frac{ \Gamma( \frac{\nu+1}{2}  )  }{ \Gamma( \frac{\nu}{2}  )  \sqrt{\nu \pi }s  } \lp(   1  +  \frac{  (x - \mu)^2 }{ \nu s^2}  \rp)^\frac{-(\nu+1)}{2} \;,\]
where $\nu(x) = \frac{\nu}{\frac{x-\mu}{s} +\nu}$, $\mathcal{I}_t (a,b)$ is the regularized incomplete Beta function, and $\Gamma(a)$ is the Gamma function.
Note that a general closed-form expression for $q_\alpha(X)$ is not known but is a readily available function within common software packages like EXCEL. 
\begin{prop} \label{prop:student_t}
If $X \sim Student$-t$(\nu,s,\mu)$, then
$$ \bar{q}_\alpha (X) =  \mu + s \lp(  \frac{\nu + T^{-1}(\alpha)^2 }   {(\nu-1)(1-\alpha)}   \rp)  \tau(T^{-1}(\alpha) )$$
where $T^{-1}(\alpha)$ is the inverse of the standardized Student$-t$ CDF and $\tau(x)$ is standardized Student$-t$ PDF. 
\end{prop}
\begin{proof}
Since there is no closed-form expression for the quantile, we utilize the representation of the superquantile given by $\frac{1}{1-\alpha} \int_{q_\alpha(X)}^\infty t f(t) dt$.
To evaluate this integral, we first take the derivative of the PDF, giving
\[ \frac{df(x)}{dx} = \frac{-f(x)(x-\mu)(\nu+1)}{\nu s^2 + (x-\mu)} .\]
Rearranging yields,
\[ xf(x)dx = \frac{-\nu s^2 df(x) }{(\nu+1)} -  \frac{(x-\mu)^2df(x)}{(\nu+1)} + \mu f(x) dx .\]
We can then integrate both sides,
\[ \int xf(x)dx = \frac{-\nu s^2 f(x) }{(\nu+1)} -  \frac{1}{(\nu+1)} \int (x-\mu)^2df(x)+ \mu F(x).\]
Integrating by parts, gives us the following form of the middle term,
\[ \int (x-\mu)^2df(x) = (x-\mu)^2 f(x) - 2 \int x f(x) dx + 2 \mu F(x) \;. \]
Then, finally, after substituting this new expression for the middle term and simplifying, we get
\[ \int xf(x)dx = - \frac{(\nu s^2 + (x- \mu)^2 )}{(\nu-1)}f(x) + \mu F(x).\]
Taking the definite integral yields,
\begin{align*}
 \int_{q_\alpha(X)}^\infty xf(x)dx &= \lp( - \lim_{x \rightarrow \infty}\frac{(\nu s^2 + (x- \mu)^2) }{(\nu-1)}f(x) +  \lim_{x \rightarrow \infty} \mu F(x)  \rp)  \\
 & \qquad  \qquad  \quad  -   \lp(- \frac{(\nu s^2 + (q_\alpha(X)- \mu)^2) }{(\nu-1)}f(q_\alpha(X)) + \mu F(q_\alpha(X))   \rp). 
 \end{align*}
It is easy to see that the second limit goes to one and, after applying L'Hopital where necessary, that the first limit goes to zero. This leaves
\begin{align*}
 \int_{q_\alpha(X)}^\infty xf(x)dx &=  \mu -   \lp(- \frac{(\nu s^2 + (q_\alpha(X)- \mu)^2) }{(\nu-1)}f(q_\alpha(X)) + \mu F(q_\alpha(X))   \rp)\\
 &= \mu( 1 - \alpha) +   \lp( \frac{\nu s^2 + (q_\alpha(X)- \mu)^2 }{(\nu-1)}  \rp)   f(q_\alpha(X))   \\
 &= \mu( 1 - \alpha)  + s \lp(  \frac{\nu + T^{-1}(\alpha)^2 }   {(\nu-1)}   \rp)  \tau(T^{-1}(\alpha) ),
 \end{align*}
 where the final step comes from writing the non-standardized quantile $q_\alpha(X)$ and PDF $f(x)$ in their standardized form. Then, finally, dividing by $1-\alpha)$ yields the formula,
 
 \[ \bar{q}_\alpha (X) =\frac{1}{1-\alpha}   \int_{q_\alpha(X)}^\infty xf(x)dx =  \mu + s \lp(  \frac{\nu + T^{-1}(\alpha)^2 }   {(\nu-1)(1-\alpha)}   \rp)  \tau(T^{-1}(\alpha) ) .\]
\qed
\end{proof}

%%%%%%%%%%%%%%%%%%%%%%%%%%%%%%%%%%%%%%%%%%%%%%%%%%%%%%%%%%%%%%%%%%%%%%%%%%%%%%%%%%%%%%%%%%%%%%%%%%%%%%%%%%%%%%%%%%%%%%%%%%%%

\subsection{Weibull}
Assume $X \sim Weibull(\lambda,k)$. Recall that Weibull parameters have range $\lambda >0$, $k>0$ with $E[X] = \lambda \Gamma(1+\frac{1}{k})$ and $\sigma^2(X)=\lambda^2\lp[ \Gamma(1+\frac{2}{k})  - \Gamma(1+\frac{1}{k})^2 \rp]$, and that the Weibull CDF, PDF, and quantile function are given by,
\[  F(x) = 1- e^{-(x/\lambda)^k} \;,\quad f(x)=\begin{cases}
\frac{k}{\lambda}\left(\frac{x}{\lambda}\right)^{k-1}e^{-(x/\lambda)^{k}} &x\geq0 ,\\
0 & x<0,
\end{cases}  \;,\] \[ \quad q_\alpha (X) = \lambda {(-\ln(1-\alpha))}^{1/k} \;.\]
where $\Gamma(a)=\int _{0}^{\infty}p^{a-1}e^{-p}dp$ is the gamma function. 

\begin{prop} \label{prop:weibull}
If $X \sim Weibull(\lambda,k)$, then
$$ \bar{q}_\alpha (X) =  \frac{\lambda}{1-\alpha} \Gamma_U \lp( 1 + \frac{1}{k} , -\ln(1-\alpha) \rp)$$
where $\Gamma_U(a,b)=\int _{b}^{\infty}p^{a-1}e^{-p}dp$ is the upper incomplete gamma function. 
\end{prop}
\begin{proof}
To calculate the superquantile, we utilize the integral representation, which is
\begin{align*}
\bar{q}_\alpha (X) &=  \frac{1}{1-\alpha} \int_{\alpha}^{1} q_p (X) dp \\
&=  \frac{1}{1-\alpha} \int_{\alpha}^{1} \lambda {(-\ln(1-p))}^{1/k} dp \;. 
\end{align*}
To put this integral into the form of the upper incomplete gamma function, make the change of variable $y= -\ln(1-p)$. This gives $e^y = \frac{1}{1-p}$ and $dp = (1-p) dy = e^{-y} dp$ with new lower limit of integration $\alpha \rightarrow -\ln(1-\alpha)$ and upper limit of integration $1 \rightarrow \infty$. Applying to the integral, yields
\begin{align*}
\bar{q}_\alpha (X) &=  \frac{\lambda}{1-\alpha} \int_{-\ln(1-\alpha)}^{\infty}  {y}^{1/k} e^{-y} dy \\
&=  \frac{\lambda}{1-\alpha} \Gamma_U \lp( 1 + \frac{1}{k} , -\ln(1-\alpha) \rp) \;. 
\end{align*}
\qed
\end{proof}
%%%%%%%%%%%%%%%%%%%%%%%%%%%%%%%%%%%%%%%%%%%%%%%%%%%%%%%%%%%%%%%%%%%%%%%%%%%%%%%%%%%%%%%%%%%%%%%%%%%%%%%%%%%%%%%%%%%%%%%%%%%%

\subsection{Log-Logistic}
Assume $X \sim LogLogistic(a,b)$. Recall that Log-Logistic parameters have range $a >0$, $b>0$ with $E[X]=a\frac{\pi}{b} csc\lp( \frac{\pi}{b} \rp) $ when $b>1$ and \\$\sigma^2(X)=a^2 \lp( \frac{2\pi}{b} csc\lp( \frac{2\pi}{b} \rp) -(\frac{\pi}{b} csc\lp( \frac{\pi}{b} \rp) )^2  \rp)$ when $b>2$, and that the Log-Logistic CDF, PDF, and quantile function are given by,
\[  F(x) = \frac{1}{1 + \lp(\frac{x}{a}\rp)^{-b} } \;,\quad f(x)={\frac  {(b /a )(x/a )^{{b -1}}}{\left(1+(x/a )^{{b }}\right)^{2}}}  \;,\quad q_\alpha (X) = a\lp( \frac{\alpha}{1-\alpha} \rp)^{\frac{1}{b}} \;,\]
where $csc(\cdot)$ is the cosecant function.

\begin{prop} \label{prop:log_logistic}

If $X \sim LogLogistic(a,b)$, then
$$ \bar{q}_\alpha (X) =\frac{a}{1-\alpha}\lp( \frac{\pi}{b} csc\lp( \frac{\pi}{b} \rp) - B_\alpha \lp( \frac{1}{b}+1, 1 -  \frac{1}{b} \rp)    \rp)$$
where $B_y(A_1,A_2)=\int _{0}^{y}p^{A_1-1}(1-p)^{A_2-1}dp$ is the incomplete beta function.
\end{prop}
\begin{proof}
To calculate the superquantile, we utilize the integral representation as follows:
\begin{align*}
\bar{q}_\alpha (X) &=  \frac{1}{1-\alpha} \int_{\alpha}^{1} q_p (X) dp \\
&=  \frac{1}{1-\alpha} \lp( \int_{0}^{1} q_p (X) dp  -  \int_{0}^{\alpha} q_p (X) dp \rp)\\
&=  \frac{1}{1-\alpha} \lp( E[X]  -   \int_{0}^{\alpha} q_p (X) dp \rp)\\
&=  \frac{1}{1-\alpha} \lp( E[X]  -  a  \int_{0}^{\alpha} \lp( \frac{p}{1-p} \rp)^{\frac{1}{b}} dp \rp) \;. 
\end{align*}
Now, note first that for $X \sim LogLogistic(a,b)$, we have $E[X] =a\frac{\pi}{b} csc\lp( \frac{\pi}{b} \rp)$. Next, for the incomplete beta function, letting $A_1 = \frac{1}{b}+1$ and $A_2=1 -  \frac{1}{b}$, we can see that
$$B_\alpha(\frac{1}{b}+1,1 -  \frac{1}{b})=\int _{0}^{\alpha}p^{\frac{1}{b}}\,(1-p)^{-\frac{1}{b}}\,dp \;.$$
Using these two facts, we have,
\begin{align*}
\bar{q}_\alpha (X) &=  \frac{1}{1-\alpha} \lp( E[X]  -  a  \int_{0}^{\alpha} \lp( \frac{p}{1-p} \rp)^{\frac{1}{b}} dp \rp) \\
&=  \frac{1}{1-\alpha} \lp( a\frac{\pi}{b} csc\lp( \frac{\pi}{b} \rp)  -  a B_\alpha(\frac{1}{b}+1,1 -  \frac{1}{b}) \rp) \\
&=  \frac{a}{1-\alpha} \lp( \frac{\pi}{b} csc\lp( \frac{\pi}{b} \rp)  -   B_\alpha(\frac{1}{b}+1,1 -  \frac{1}{b}) \rp) \;.
\end{align*}
\qed
\end{proof}

%\begin{cor} \label{cor:log_logistic}
%If $X \sim LogLogistic(a,b)$, then
%$$ \bar{q}_\alpha (X) =\frac{a}{1-\alpha}\lp( \frac{\pi}{b} csc\lp( \frac{\pi}{b} \rp) -  \frac{ b\alpha^{\frac{1}{b}+1}}{b+1} {}_2F_1(\frac{1}{b},\frac{1}{b}+1;\frac{1}{b}+2;\alpha)    \rp)$$
%where ${}_2F_1(c_1,c_2;c_3;C)$ is the hypergeometric function and $csc(\cdot)$ is the cosecant function.\end{cor}
%\begin{proof}
%This simply utilizes an alternative representation of the incomplete beta function given by $\frac{ b\alpha^{\frac{1}{b}+1}}{b+1} {}_2F_1(\frac{1}{b},\frac{1}{b}+1;\frac{1}{b}+2;\alpha) $. This alternate representation using the hypergeometric function is well known. (see e.g. http://mathworld.wolfram.com/IncompleteBetaFunction.html)
%\end{proof}
%%%%%%%%%%%%%%%%%%%%%%%%%%%%%%%%%%%%%%%%%%%%%%%%%%%%%%%%%%%%%%%%%%%%%%%%%%%%%%%%%%%%%%%%%%%%%%%%%%%%%%%%%%%%%%%%%%%%%%%%%%%%

\subsection{Generalized Extreme Value Distribution}
Assume $X$ follows a Generalized Extreme Value (GEV) Distribution, which we denote as $X \sim GEV(\mu,s,\xi)$. Recall that GEV parameters have range $\mu \in \mathbb{R}$, $s>0$, $\xi \in \mathbb{R}$ with $E[X]=\begin{cases}\mu +s (g_{1}-1)/\xi &{\text{if}}\ \xi \neq 0,\xi <1,\\\mu +s \,y &{\text{if}}\ \xi =0,\\\infty &{\text{if}}\ \xi \geq 1,\end{cases}$ and \\$\sigma^2(X)= \begin{cases}s^2\,(g_2-g_1^2)/\xi^2 & \text{if}\ \xi\neq0,\xi<\frac12,\\ s^2\,\frac{\pi^2}{6} & \text{if}\ \xi=0, \\ \infty & \text{if}\ \xi\geq\frac12,\end{cases}$ where $g_k = \Gamma(1- k\xi)$ and $y$ is the Euler-Mascheroni constant.

Additionally, recall that the GEV has CDF, PDF, and quantile function given by,
\[  F(x) =\begin{cases}
e^{-\lp(  1 +\frac{ \xi(x-\mu) }{s} \rp)^\frac{-1}{\xi} }& \xi \neq 0  ,\\
e^{ -e^{-\lp( \frac{x-\mu}{s}  \rp) }}& \xi =0
\end{cases} 
\;,\].   \[\quad f(x)=\begin{cases}
\frac{1}{s}\lp(  1 +\frac{ \xi(x-\mu) }{s} \rp)^{ \frac{-1}{\xi} -1  }e^{-\lp(  1 +\frac{ \xi(x-\mu) }{s} \rp)^\frac{-1}{\xi} }& \xi \neq 0  ,\\
\frac{1}{s} e^{-\lp( \frac{x-\mu}{s}  \rp) } e^{ -e^{-\lp( \frac{x-\mu}{s}  \rp) }}& \xi =0,
\end{cases}  \;
\]
\[
\quad q_\alpha (X)= \begin{cases}
\mu + \frac{s}{\xi} \lp( (\ln(\frac{1}{\alpha})^{-\xi} -1  \rp)& \xi \neq 0  ,\\
\mu - s \ln(-\ln(\alpha)) & \xi =0.
\end{cases} \;\]

\begin{prop}
If $X \sim GEV(\mu,s,\xi)$, then
$$ \bar{q}_\alpha (X) =  \begin{cases}
\mu + \frac{s}{\xi(1-\alpha)} \lp[  \Gamma_L(1-\xi,\ln(\frac{1}{\alpha}))  - (1-\alpha) \rp]& \xi \neq 0  ,\\
\mu +  \frac{s}{(1-\alpha)} \lp(  y+ \alpha\ln(-\ln(\alpha)) - li(\alpha) \rp) & \xi =0,
\end{cases}$$
where $\Gamma_L (a,b)=\int _{0}^{b}p^{a-1}e^{-p}dp$ is the lower incomplete gamma function, $li(x) =\int _{0}^{\alpha} \frac{1}{\ln p} dp$ is the logarithmic integral function, and $y$ is the Euler-Mascheroni constant. 
\end{prop}
\begin{proof}
Assume we have $\xi=0$. Then, we have
\begin{align*}
\bar{q}_\alpha (X) &=  \frac{1}{1-\alpha} \int_{\alpha}^{1} \mu - s \ln(-\ln(p))dp \\
&=  \mu - \frac{s}{1-\alpha} \lp( \int_{0}^{1} \ln(-\ln(p))dp  -  \int_{0}^{\alpha} \ln(-\ln(p))dp \rp)\\
&=  \mu - \frac{s}{1-\alpha} \lp(y  -  \int_{0}^{\alpha} \ln(-\ln(p))dp \rp)\\
&=  \mu - \frac{s}{1-\alpha} \lp(y  + \alpha\ln(-\ln(\alpha)) - li(\alpha) \rp)\\
\end{align*}
Assume now that $\xi \neq 0$. Then, we have that,
\begin{align*}
\bar{q}_\alpha (X) &=  \frac{1}{1-\alpha} \int_{\alpha}^{1} \mu + \frac{s}{\xi} \lp( (\ln(\frac{1}{p})^{-\xi} -1  \rp)dp \\
&=  \mu - \frac{s}{\xi(1-\alpha)} \int_{\alpha}^{1} \lp( (\ln(\frac{1}{p})^{-\xi} -1  \rp)dp\\
&=  \mu + \frac{s}{\xi(1-\alpha)} \lp[  \Gamma_L(1-\xi,\ln(\frac{1}{\alpha}))  - (1-\alpha) \rp]\\
\end{align*}
\qed
\end{proof}

\section{Portfolio Optimization}
A common parametric approach to portfolio optimization is to assume that portfolio returns follow some specified distribution. In this context, particularly when taking a risk averse approach, closed-form representations of the superquantile and bPOE for the specified distribution allow one to formulate a tractable portfolio optimization problem.  In this section, we show that our derived formulas for the superquantile and bPOE reveal important properties about portfolio optimization problems formulated with particular distributional assumptions placed upon portfolio returns.  

Portfolio optimization with the superquantile is common, so we begin by simply pointing out which of the closed-form superquantile formulas yield tractable portfolio optimization problems. Portfolio optimization with bPOE, however, is not common and we show that it can be advantageous compared to the superquantile approach. In particular, superquantile optimization requires that one sets the probability level $\alpha$. One can then observe that for fixed $\alpha$, the optimal superquantile portfolio may change based upon the distribution utilized to model returns. We show that if portfolio returns are assumed to follow a Laplace, Logistic, Normal, or Student-$t$ distribution, the minimal bPOE portfolio's for fixed threshold $x$ are the same regardless of the distribution chosen, meaning that there exists a single portfolio that is $x$-bPOE optimal for multiple choices of distribution.

Note that in this section we will be dealing with asset returns $R$, as it is typical for financial problems, and the loss is the opposite of return: $X = -R$, and $q_\alpha(X) = -q_{1-\alpha}(R)$. 

The portfolio optimization problem consists of finding a vector of asset weights $w \in \mathbb{R}^n$ for a set of $n$ assets with unknown random returns $R =[ R_1, R_2,...,R_n]$ that solves the following optimization problem,

\begin{equation}
\label{port_standard}
\begin{aligned}
& \underset{ w \in \mathbb{R}^n}{ \text{max }}  
& & L(w,R) \\
& s.t. 
& & g_i(w,R) \leq 0 , i=1,...,I\\
& & & h_j(w,R) =0 , j=1,...,J\\
& & & w^T 1 = 1 \\
& & & l \leq w \leq u
\end{aligned}
\end{equation}
where $L(w,R)$ is some function to be maximized,\footnote{Or minimized if we consider the negative.} functions $g_j(w,R)$ and $h_i(w,R)$ enforce inequality and equality constraints respectively, and vectors $l,u$ enforce upper and lower bounds on the individual asset weights. A simple example is the standard Markowitz optimization problem where we maximize the expected utility, which is a weighted combination of the expected return and its variance via a positive trade-off parameter $\lambda \geq 0$:

\begin{equation}
\label{mark_1}
\begin{aligned}
& \underset{ w \in \mathbb{R}^n}{ \text{max }}  
& & w^T\eta - \lambda w^T \Sigma w \\
& s.t. 
& & w^T 1 = 1 \\
& & & l \leq w \leq u
\end{aligned}
\end{equation}

An important aspect of the random portfolio return $w^T R$ which can be seen within the Markowitz problem and will be used later on in this section is the fact that the expectation $E[w^TR]$ and variance $\sigma^2(w^TR)$ are given by $w^T \eta$ and $w^T \Sigma w$ respectively, where $\eta \in \mathbb{R}^n$ is the vector of expected returns for the $n$ assets and $\Sigma \in \mathbb{R}^{n \times n}$ is the covariance matrix for the $n$ assets. This allows us to represent the expected value and variance of the portfolio return in terms of $w$, and consequently to formulate an optimization problem with decision vector $w$. 

\subsection{Superquantile and bPOE Optimization with Qualified Distributions}

As we are dealing with asset returns, and not losses, we need to define the superquantile using that notation. The superquantile is the expected loss above the quantile (conditional expected value of losses in the right tail), so in terms of returns it would be the conditional expected value of returns in the left tail, which can be described by the left superquantile:

$$ -\tilde{q}_{1-\alpha} (R) = \frac{1}{1-\alpha} \int_{0}^{1-\alpha} q_p (R) dp. $$

We can use the closed-form superquantile formulas derived in the previous sections for the right superquantile $\bar{q}_\alpha (R)$ to calculate the left superquantile $\tilde{q}_\alpha (R)$, as

$$ \alpha \tilde{q}_\alpha(R) + (1 - \alpha) \bar{q}_\alpha(R) = \int_{0}^{1} q_p(R) dp = E[R], $$
so
$$  -\tilde{q}_{1-\alpha} (R) = -\frac{1}{1-\alpha} (E[R] - {\alpha} \bar{q}_{1-\alpha}(R)). $$

Since $-\tilde{q}_{1-\alpha} (R) = \bar{q}_\alpha (X)$, bPOE is defined as $$ \bar{p}_x (X) =  \{ 1 - \alpha | \bar{q}_\alpha (X) = x \} = \{ 1 - \alpha | \tilde{q}_{1-\alpha} (R) = -x \}.$$

\subsubsection{Qualified Distributions for Portfolio Optimization}
The superquantile or bPOE portfolio optimization problem has its objective function or one of the constraints defined in terms of $\tilde{q}_{1-\alpha}(w^TR)$ or $\bar{p}_x (w^T R)$. To formulate such a problem using a given distribution, we begin by defining a set of \textit{qualified} distribution's which we will consider. These qualified distribution's satisfy the following set of conditions, which allow us to verify that they make sense in terms of portfolio theory and admit superquantile/bPOE expression in a way that can represented in terms of decision variable $w$:

\begin{definition}[Qualified Distribution] A \textit{qualified} distribution $\mathcal{D}$ satisfies the following conditions:\\
(C1)  $w^TR\sim \mathcal{D} \implies \tilde{q}_{1-\alpha} (w^T R) = w^T \eta - \sqrt{w^T \Sigma w} \zeta(\alpha, \Theta)  $, where $\zeta(\alpha, \Theta)$ is a function depending only upon $\alpha$ and possibly a set of fixed parameters $\Theta$ that do not depend on $w$, $\eta$ is the vector of the expected asset returns, and $\Sigma$ is the covariance matrix for asset returns. \\
(C2) The statistical parameters of the distribution $\mathcal{D}$ must be consistent with the descriptive statistics of real-life asset returns. \\
(C3) The shape of the PDF for the given distribution $\mathcal{D}$ must conform to the shape of the empirical PDF of typical real-life asset returns.
\end{definition}

Why should we enforce these preconditions? Condition (C1) guarantees that the superquantile can be expressed in terms of $w$. This is necessary to express the superquantile optimization problem. For example, if we assume that $w^TR \sim Logistic(\mu,s)$, we need to be able to express $\mu$ and $s$ in terms of $w$. Since $\mu = E[w^T R] = w^T \eta$ and $w^T \Sigma w=\sigma^2(w^TR)=\frac{s^2 \pi^2}{3}$, we have
\begin{align*}
\tilde{q}_{1-\alpha} (R) &= \frac{1}{1-\alpha} (E[R] - {\alpha} \bar{q}_{1-\alpha}(R)) = \frac{1}{1-\alpha}(\mu - \alpha [\mu + \frac{s}{\alpha}(-(1-\alpha)ln(1-\alpha)-\alpha ln(\alpha)) ]) \\
&= \mu - \frac{s}{1-\alpha}(-\alpha ln(\alpha) - (1-\alpha)ln(1-\alpha)) \\
&= w^T \eta - \sqrt{w^T \Sigma w} \frac{\sqrt{3}(-\alpha \ln(\alpha) - (1-\alpha) \ln (1-\alpha))}{\pi(1-\alpha)} \; ,
\end{align*}
which satisfies (C1). Other examples that satisfy this condition are Laplace, Normal, Exponential, Student$-t$, Pareto, GPD, and GEV. Note that for Student$-t$, we assume that parameter $\nu$ is fixed and the same for all assets, i.e. $\Theta = \{\nu \} $, and for GPD/GEV distribution's $\Theta = \{ \xi \}$.

Condition (C2) and (C3) are simple sanity checks on our model assumptions. For example, for exponential distribution $E[R]=\frac{1}{\lambda}=\sigma(R)$, however for the real-life asset returns the sample mean is not generally equal to the sample standard deviation. So, Exponential and Pareto distribution's make no sense in portfolio optimization problems even if they satisfy (C1). As for (C3), a distribution is not practical if there is obvious discrepancy between the shape of its PDF and the shape of the empirical PDF observed using real-life asset returns. The latter is generally bell-shaped or, more likely, inverse-V shaped, and is never shaped like the PDF of an Exponential, Pareto/GPD, or Weibull for $k < 1$.

This leaves us with a set of four elliptical distribution's which satisfy all three conditions: Logistic, Laplace, Normal, and Student$-t$, as well as the nonelliptical GEV distribution. For the latter, with $\xi \neq 0$ the left superquantile can be expressed as
$$ \tilde{q}_{1-\alpha} (R) = w^T \nu - \sqrt{w^T \Sigma w} \frac{ \alpha \Gamma (1 - \xi) - \Gamma_U (1 - \xi, ln(\frac{1}{1-\alpha})) }{ (1 - \alpha) \sqrt{g_2 - g_1^2} }, $$
where $\Gamma_U (a,b)=\int _{b}^{\infty}p^{a-1}e^{-p}dp$ is the upper incomplete Gamma function, $g_k = \Gamma (1 - k \xi)$.

\subsubsection{Superquantile and bPOE Optimization}
An alternative to the Markowitz problem is to find the portfolio with minimal superquantile (\ref{min_CVAR}) or bPOE (\ref{min_bPOE}).

\begin{equation}
\label{min_CVAR}
\begin{aligned}
& \underset{ w \in \mathbb{R}^n}{ \text{min }}  
& & -\tilde{q}_{1-\alpha} (w^T R) \\
& s.t. 
& &  w^T 1 = 1 \\
& & & l \leq w \leq u
\end{aligned}
\end{equation}

\begin{equation}
\label{min_bPOE}
\begin{aligned}
& \underset{ w \in \mathbb{R}^n}{ \text{min }}  
& & \bar{p}_x (w^T R) \\
& s.t. 
& &  w^T 1 = 1 \\
& & & l \leq w \leq u
\end{aligned}
\end{equation}

For qualified distribution's, however, these problems can be greatly simplified. First, we see that (\ref{min_CVAR}) reduces to (\ref{min_CVAR2}):
\begin{equation}
\label{min_CVAR2}
\begin{aligned}
& \underset{ w \in \mathbb{R}^n}{ \text{max }}  
& & w^T \eta - \sqrt{w^T\Sigma w} \zeta(\alpha, \Theta) \\
& s.t. 
& &  w^T 1 = 1 \\
& & & l \leq w \leq u
\end{aligned}
\end{equation}

\cite{Khokhlov} shows that the optimal solution to (\ref{min_CVAR2}) is the same as the optimal solution to a Markowitz optimization problem (\ref{mark_1}) with $\lambda = \frac{\zeta(\alpha, \Theta )}{2\sigma(w^T R)}$. Thus, the superquantile optimal portfolio is also mean-variance optimal in the Markowitz sense.

Now, for bPOE we see that the picture is actually much simpler. Specifically, we have the following proposition.
\begin{prop} If we assume that $w^T R \sim \mathcal{D}$ and we have that $\mathcal{D}$ is a qualified distribution, then (\ref{min_bPOE}) reduces to (\ref{max_Sharpe}).

\begin{equation}
\label{max_Sharpe}
\begin{aligned}
& \underset{ w \in \mathbb{R}^n}{ \text{max }}  
& & \frac{w^T \eta + x}{ \sqrt{w^T \Sigma w } } \\
& s.t. 
& &  w^T 1 = 1 \\
& & & l \leq w \leq u
\end{aligned}
\end{equation}
\end{prop}

 \begin{proof}
 First, note that by definition of the superquantile, we know that $\zeta(\alpha,\Theta)$ must be an increasing function w.r.t. $\alpha \in [0,1]$. Second, as $ \bar{p}_x (X) = \{ 1 - \alpha | \tilde{q}_{1-\alpha} (R) = -x \}$ and $\tilde{q}_{1-\alpha} (w^T R) = w^T \eta - \sqrt{w^T \Sigma w} \zeta(\alpha, \Theta)$ for qualified distribution's, the problem (\ref{min_bPOE}) can be rewritten as:
 \begin{equation}
\label{min_CVAR3}
\begin{aligned}
& \underset{ w \in \mathbb{R}^n,\alpha}{ \text{min }}  
& & 1-\alpha \\
& s.t. 
& &  -w^T \eta + \sqrt{w^T\Sigma w} \zeta(\alpha,\Theta ) = x \\
& & & w^T 1 = 1 \\
& & & l \leq w \leq u
\end{aligned}
\end{equation}
which can then be written as:
  \begin{equation}
\label{min_CVAR4}
\begin{aligned}
& \underset{ w \in \mathbb{R}^n}{ \text{max }}  
& & \alpha \\
& s.t. 
& &    \zeta(\alpha,\Theta ) = \frac{w^T \eta + x}{\sqrt{w^T\Sigma w}} \\
& & & w^T 1 = 1 \\
& & & l \leq w \leq u
\end{aligned}
\end{equation}
 Finally, since $\zeta(\alpha,\Theta )$ is an increasing function w.r.t. $\alpha$ and $\Theta$ does not dependent upon $w$, we see that we can formulate the maximization as (\ref{max_Sharpe}) without changing the argmin. 
 \qed
  \end{proof}

This proposition has the important implication for portfolio theory that the optimal bPOE portfolio for the qualified distribution does not depend on the distribution itself. The same portfolio will have the lowest bPOE for any of those distribution's. The fact that bPOE optimization is, in some sense, independent from distributional assumptions makes it preferable to superquantile optimization.

\subsubsection{Numerical Demonstration}
In this example, we consider a global equity portfolio that consists of 6 market portfolios - U.S., Japan, U.K., Germany, France, and Switzerland, represented by the corresponding MSCI indices - MXUS, MXJP, MXGB, MXDE, MXFR, MXCH. Parameters of returns for portfolio components are provided in Table~\ref{table:1} (source: Capital IQ sample of monthly returns from April 1987 to April 1996, annualized).

\begin{table}
\caption{Portfolio Return Data}
\label{table:1}

\resizebox{15.9cm}{!} {
\begin{tabular}{|l|r r|r r r r r r|}
\hline
Asset &Expected &Standard               &\multicolumn{6}{|l|}{Correlations} \\
ticker&return   &deviation&MXUS		    &MXJP		&MXGB		&MXDE		&MXFR		&MXCH \\ 
\hline
MXUS  &10.25\%	&13.79\%	&1		    &0.190041	&0.639133	&0.481857	&0.499406	&0.605384\\
MXJP  &6.90\%	&26.05\%	&0.190041	&1			&0.450337	&0.251601	&0.378753	&0.373964\\
MXGB  &8.81\%	&19.16\%	&0.639133	&0.450337	&1			&0.579918	&0.584215	&0.654687\\
MXDE  &9.15\%	&20.31\%	&0.481857	&0.251601	&0.579918	&1			&0.753072	&0.628426\\
MXFR  &8.83\%	&20.40\%	&0.499406	&0.378753	&0.584215	&0.753072	&1			&0.580626\\
MXCH  &13.85\%	&17.45\%	&0.605384	&0.373964	&0.654687	&0.628426	&0.580626	&1\\
\hline
\end{tabular}
}
\end{table}
This problem was solved using a non-linear programming algorithm and the results are provided in Table~\ref{table:2}. The respective values of $\lambda$ are also provided, which allows deriving the same portfolios using the standard MVO solver that uses a quadratic programming algorithm.
\begin{table}
\caption{Optimal Superquantile (CVaR) Portfolio's}
\label{table:2}

\newcolumntype{d}[1]{D{.}{.}{#1} }
\resizebox{15.9cm}{!} {
\begin{tabular}{|l|r|r r r r|r r r r|}
\hline
Asset   &\multicolumn{1}{|c|}{min risk} &	\multicolumn{4}{|c|}{CVaR 99\% optimal portfolios} & \multicolumn{4}{|c|}{CVaR 95\% optimal portfolios} \\
ticker  &\multicolumn{1}{|c|}{portfolio}&normal &t (df=3)&Laplace	&logistic &normal  &t (df=3) &Laplace &logistic\\
\hline
MXUS    &70.99\% &65.80\% &67.59\% &67.03\%	&66.53\%  &64.23\% &64.78\%	 &65.05\% &64.64\% \\
MXJP    &13.98\% &9.61\%  &11.11\% &10.64\%	&10.21\%  &8.28\%  &8.74\%	 &8.97\%  &8.62\%  \\
MXGB    &0.00\%	 &0.00\%  &0.00\%  &0.00\%	&0.00\%	  &0.00\%  &0.00\%	 &0.00\%  &0.00\%  \\
MXDE    &9.24\%	 &2.87\%  &5.07\%  &4.37\%	&3.76\%	  &0.95\%  &1.61\%	 &1.94\%  &1.44\%  \\
MXFR    &0.00\%	 &0.00\%  &0.00\%  &0.00\%	&0.00\%	  &0.00\%  &0.00\%	 &0.00\%  &0.00\%  \\
MXCH    &5.79\%	 &21.72\% &16.22\% &17.96\%	&19.50\%  &26.54\% &24.87\%	 &24.04\% &25.30\% \\
\hline
Return  &9.89\%	 &10.68\% &10.40\% &10.49\%	&10.57\%  &10.91\% &10.83\%	 &10.79\% &10.85\% \\
St.dev. &12.86\% &13.01\% &12.93\% &12.95\%	&12.97\%  &13.11\% &13.08\%	 &13.06\% &13.09\% \\
$\lambda$ &        &20.48   &31.28   &26.82   &23.80    &15.73   &17.11    &17.88	  &16.73   \\
\hline
\end{tabular}
}
\end{table}
Table~\ref{table:2} shows that superquantile optimal portfolios are not the same as the global minimum variance portfolio (min. risk portfolio), but are quite close to it. Distributional assumptions play their role in the optimal portfolios composition, with the Student-t distribution rendering the most conservative allocation for CVaR 99\%. However, the differences between optimal portfolios for CVaR 95\% are insignificant.

We can note from (\ref{min_CVAR2}) that if portfolio return is constrained from below, unless this constraint is very close to the return of the global minimum variance portfolio, it results in the superquantile optimization being essentially the same as the variance minimization. If the risk is constrained from above, that superquantile optimization is the same as return maximization. 

Using the same set of assets, we also solved the bPOE optimization problem (\ref{min_bPOE}) with thresholds $x = 0.16$ and $x = 0.25$ (i.e. losses exceeding 16\% and 25\% of the initial portfolio respectively), $l=0$, and $u=1$. Table~\ref{table:3} shows the results: the minimal bPOE achieved, the optimal portfolio composition and parameters, and CVaR for the optimal portfolios for all distribution's.
\begin{table}
\caption{Optimal bPOE Portfolio's}
\label{table:3}

\resizebox{15.9cm}{!} {
\begin{tabular}{|l|r r r r|r r r r|}
\hline
Assumed distribution& normal & t (df=3) & Laplace & logistic & normal  & t (df=3) & Laplace & logistic\\
\hline
bPOE threshold, $x$& \multicolumn{4}{|c|}{16\%}	          & \multicolumn{4}{|c|}{25\%} \\
bPOE value, $\bar{p}_x(X)$& 5.13\% & 6.21\%   & 7.46\%  & 6.36\%   & 0.80\%  & 2.93\%   & 2.81\%  & 1.86\%\\
\hline
Asset tiker      & \multicolumn{8}{|c|}{bPOE-optimal portfolio composition} \\
MXUS	         & 64.20\%&	64.19\%	 & 64.20\% & 64.20\%  & 65.95\% & 65.95\%  & 65.95\% &65.95\%\\
MXJP	         & 8.26\% & 8.27\%	 & 8.25\%  & 8.25\%	  & 9.73\%	& 9.73\%   & 9.73\%	 & 9.73\%\\
MXGB             & 0.00\% & 0.00\%   & 0.00\%  & 0.00\%   & 0.00\%  & 0.00\%   & 0.00\%  & 0.00\%\\
MXDE             & 0.90\% & 0.91\%   & 0.90\%  & 0.90\%   & 3.05\%  & 3.05\%   & 3.06\%  & 3.05\%\\
MXFR             & 0.00\% & 0.00\%   & 0.00\%  & 0.00\%   & 0.00\%  & 0.00\%   & 0.00\%  & 0.00\%\\
MXCH             & 26.64\%& 26.63\%  & 26.64\% & 26.65\%  & 21.27\% & 21.27\%  & 21.27\% & 21.27\%\\
\hline
Return           & 10.92\%&	10.92\%	 & 10.92\% & 10.92\%  & 10.65\% & 10.65\%  & 10.65\% & 10.65\%\\
St.dev.	         & 13.12\%&	13.12\%	 & 13.12\% & 13.12\%  & 13.00\% & 13.00\%  & 13.00\% & 13.00\%\\
\hline
Test distribution& \multicolumn{8}{|c|}{CVaR for the test distribution's} \\
normal	         & 16.00\%&	14.93\%	& 13.87\% & 14.79\%   & 25.00\%	& 18.95\%  & 19.16\% & 21.16\%\\
t (df=3)	     & 18.14\%&	16.00\%	& 14.05\% & 15.74\%   & 46.31\%	& 25.00\%  & 25.56\% & 31.46\%\\
Laplace	         & 19.48\%&	17.70\%	& 16.00\% & 17.48\%   & 36.62\%	& 24.61\%  & 25.00\% & 28.79\%\\
logistic	     & 17.61\%&	16.18\%	& 14.81\% & 16.00\%   & 31.14\%	& 21.71\%  & 22.01\% & 25.00\%\\
\hline
\end{tabular}
}
\end{table}
All the optimal solutions were generated by the non-linear programming algorithm for different distribution's, but the results support the conclusion that the optimal portfolio composition does not depend on the distribution (small discrepancies are due to the optimization algorithm accuracy).

\section{Parametric Density Estimation with Superquantile's}
One of the motivation's for providing closed-form superquantile formulas is so that they can be used within common parametric estimation frameworks. The Exponential, Parteo/GPD, Laplace, Normal, LogNormal, Logistic, Student-t, Weibull, LogLogistic, and GEV represent a wide range of distribution's that can now be utilized within these parametric procedures, but with superquantile's incorporated into the fitting criteria. We illustrate this idea by proposing a simple variation of the Method of Moments (MM), which we call the Method of Superquantile's (MOS), where superquantile's at varying levels of $\alpha$ take the place of moments. Our numerical example utilizes a heavy tailed Weibull to illustrate MOS, since it is particularly well-suited for asymmetric heavy-tailed data. However, any of the listed distribution's could be used as well. 

\subsection{Method of Superquantile's}
The MM is a well known tool for estimating the parameters of a distribution when moments are available in parametric form and desired moments are either assumed to be known or are measured from empirical observations. It looks for the distribution $f_\Theta(x)$, parameterized by $\Theta$, with moments equal to some known moments or, if moments are unknown, empirical moments. With $n$ moments used, the problem reduces to solving a set of $n$ equations w.r.t. the set of parameters $\Theta$ of the distribution family. 

This method, though, can be generalized where moments are replaced by other distributional characteristics, such as the superquantile and quantile. We utilize superquantile's in this context. This method provides flexibility through the choices of different $\alpha$, allowing the user to focus the fitting procedure on particular portions of the distribution. This flexibility is advantageous compared to other methods such as MM or maximum likelihood (ML), since these fitting methods treat each portion of the distribution equally. When fitting the tail is important, for example, and there are many samples around the mean with few samples in the tail, it can be desirable to focus the fitting procedure on carefully fitting the tail samples. As will be shown, one can focus MOS by choice of $\alpha$. One will see that this procedure is similar to fitting with Probability Weighted Moments (PWM)\footnote{also sometimes called L-moments.}, but where MOS is much more straightforward with superquantiles far easier to interpret than PWM's.

We formulate the following problem, where $\hat{\bar{q}}_\alpha(X)$ denotes either a known superquantile or an empirical estimate from a sample of $X$ and $\bar{q}_\alpha(X_{f_\Theta})$ denotes parameterized formulas for the superquantile when $X$ has density function $f_\Theta$ with the set of parameters $\Theta$:\\

\noindent \textbf{Method of Superquantiles}:
Fix $\alpha_1,...,\alpha_k \in [0,1]$ and choose a parametric distribution family $f_\Theta$ with parameters $\Theta$. Solve for $\Theta$ such that,
\[ \bar{q}_{\alpha_i} (X_{f_\Theta}) = \hat{\bar{q}}_{\alpha_i}(X) \text{ for all } i=1,...,k,\]
which is a system of $k$ equations in $|\Theta|$ unknowns.

This problem, however, may not have a solution. For example, if $k=2$ and the parametric family only has a single parameter (i.e. $|\Theta|=1$). In this case, one can solve the following surrogate Least Squares minimization problem:\\

\noindent \textbf{LS Method of Superquantiles (LS-MOS)}:
Fix $\alpha_1,...,\alpha_k \in [0,1]$ and choose a parametric distribution family $f_\Theta$ with parameters $\Theta$. Choose weights $c_1,...,c_k>0$ and solve for,
\[ \Theta \in \underset{\Theta}{\text{argmin }} \sum_i c_i \lp(\bar{q}_{\alpha_i} (X_{f_\Theta}) - \hat{\bar{q}}_{\alpha_i}(X) \rp)^2 \;. \]
This procedure finds the distribution which has superquantile's that are \textit{close} to the empirical superquantile's. The freedom to select $\alpha_i$ as well as $c_i$ provides the user with much flexibility as to which portion of the distribution should match more exactly the empirical superquantile's. 

\subsubsection{Example Customization: Conservative Tail Fitting}
When sample size is small and the tail of the distribution at hand is long, it is likely that the tail will be difficult to characterize from empirical data since few observations will be observed in the tail (with high probability). The proposed method of superquantile's can easily, however, be made more conservative based upon empirical data in an intuitive way. For example, one could have the following condition where $\epsilon_i$ is a pre-specified constant such that $0<\epsilon_i\leq \alpha_i$:
\[ \bar{q}_{\alpha_i-\epsilon_i} (X_{f_\Theta}) = \hat{\bar{q}}_{\alpha_i}(X).\]
Or, for the least squares variant, one can fit the problem,

\[ \min_\Theta \sum_i c_i (\bar{q}_{\alpha_i-\epsilon_i} (X_{f_\Theta}) - \hat{\bar{q}}_{\alpha_i}(X))^2 \]

Notice that these conditions are effectively making the assumption that the empirical superquantile has underestimated the true tail expectation, which is often the case with heavy tailed distribution's. 

\subsubsection{Example: Weibull Distribution Fitting}
We illustrate the basic method on fitting a Weibull distribution, with $\Theta=(\lambda,k)$, from a small sample of 50 observations. We took two independent samples, denoted $S_1,S_2$, of size 50 from a Weibull with $\lambda=.5, k=1.4$. We then estimated the Weibull parameters using MM, ML, and the LS-MOS. The MM was solved using the first two moments. The LS-MOS was solved twice. It was first solved with $\alpha_1=.15, \alpha_2=.75, c_1=c_2=1$, a choice which was made to mimic the behavior of MM and ML, where the fit emphasizes most of the observed data. To put more emphasis on the tail observations, it was also solved with $\alpha_1=.5, \alpha_2=.75, \alpha_3=.95, c_1=c_2=c_3=1$. We denote these solutions as LS1, LS2 respectively. The ML solution is available in closed-form and we solved MM, LS1, and LS2 using Scipy's optimization library.\footnote{Specifically, we used the \textit{leastsq} function which implements MINPACK's \textit{lmdif} routine. This routine requires function values and calculates the Jacobian by a forward-difference approximation.}

Looking at Figure~\ref{fig:S1} for $S_1$ and Figure~\ref{fig:S2} for $S_2$, we see that the LS1 fit is, indeed, much like the MM a ML fit for both data sets. However, we see that the LS2 fit is the best in both cases. The ML, MM, and LS1 methods have put too much emphasis on the observations around the mode. The LS2 fit, however, has put appropriate emphasis on the less frequent observations in the tail. 

It is also important to notice how the differences in $S_1$ and $S_2$ have affected the fit from each method. Looking at the differences between $S_1$ and $S_2$, we can see that the samples differ in the observed density in the lower portion of the range. This is directly reflected in the fits given by MM, ML, and LS1. Compared to their fits on $S_1$, they are more heavily favoring the left side of the distribution. The LS2 fit, however, is robust to these differences between data sets and, by focusing on the tail, has remained mostly unchanged from the fit on $S_1$. This is the intended effect from the selection of larger values for $\alpha$ in LS2.

We duplicated this procedure on a heavier tailed Weibull. We took 50 samples of a Weibull with true parameters $k=1 , \lambda=.5$ and fit MM, ML, LS1, and LS2 using the empirical data. Figure~\ref{fig:H1} and Figure~\ref{fig:H2} highlight different aspects of the resulting fits. We see that LS2 clearly provides the best fit, with Figure~\ref{fig:H2} in particular showing that MM, ML, and LS1 underestimate the tail densities. MM, ML, and LS1 put more emphasis on fitting the observations around the mode. As intended, however, LS2 focuses more on fitting the right tail observations and arrives at a better fit.

\begin{figure}
\label{none}
\resizebox{1\linewidth}{!} {
%\centering
\begin{minipage}{.5\textwidth}
%  \centering
  \captionsetup{width=.9\linewidth}
  \includegraphics[width=3in,height=2in]{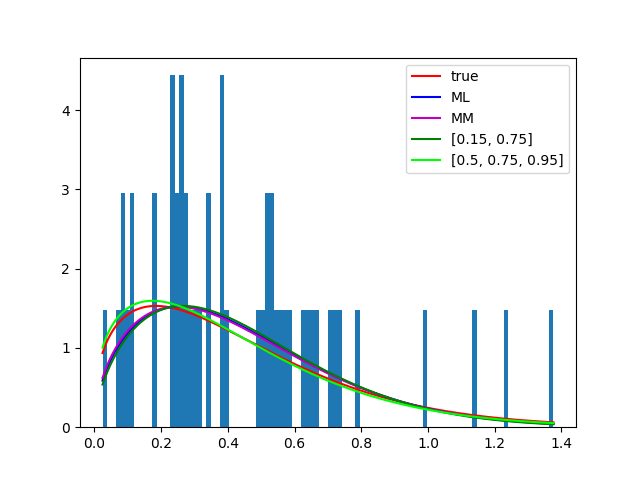} 
 \captionof{figure}{ Fits using sample $S_1$. PDF's displayed with normalized histogram of $S_1$ sample given in background.  }
  \label{fig:S1}

\end{minipage}%
\begin{minipage}{.5\textwidth}

%  \centering
  \captionsetup{width=.9\linewidth}
  \includegraphics[width=3in,height=2in]{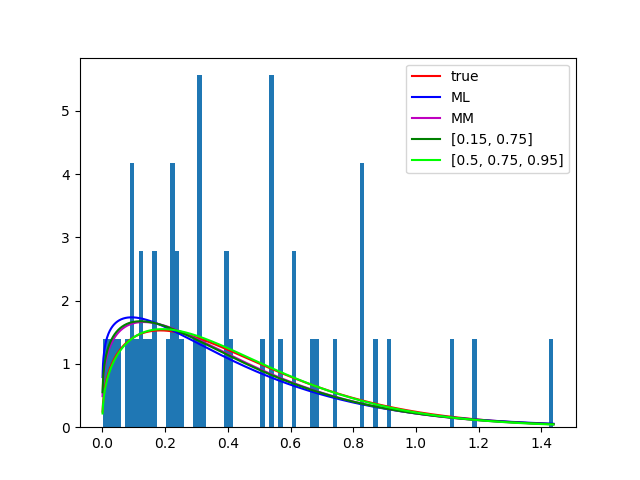} 
 \captionof{figure}{ Fits using sample $S_2$. PDF's displayed with normalized histogram of $S_2$ sample given in background.  }
  \label{fig:S2}

\end{minipage}

}
\end{figure}

\begin{figure}
\label{none}
\resizebox{1\linewidth}{!}{
%\centering
\begin{minipage}{.5\textwidth}
 % \centering
  \captionsetup{width=.9\linewidth}
  \includegraphics[width=3in,height=2in]{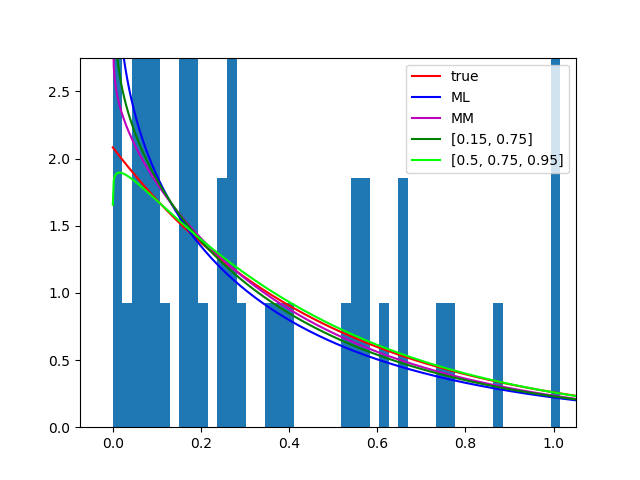} 
 \captionof{figure}{ Left side of distribution for fits on sample from Weibull with true $k=1 , \lambda=.5$. }
  \label{fig:H1}

\end{minipage}%
\begin{minipage}{.5\textwidth}

%  \centering
  \captionsetup{width=.9\linewidth}
  \includegraphics[width=3in,height=2in]{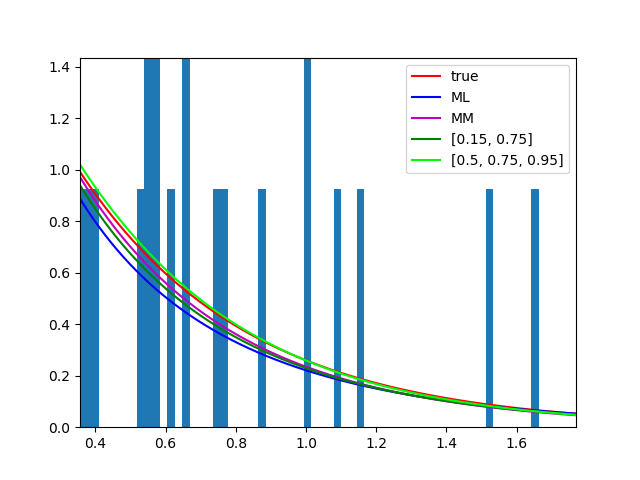} 
 \captionof{figure}{ Right side of distribution for fits on sample from Weibull with true $k=1 , \lambda=.5$.  }
  \label{fig:H2}

\end{minipage}
}
\end{figure}

\subsubsection{Constrained Likelihood and Entropy Maximization}
While we focused primarily on a variant of the method of moments, the formulas provided for superquantile's and bPOE can be used in other parametric procedures. For example, one could consider a constrained variant of the maximum likelihood or maximum entropy method, where superquantile constraints are introduced. Letting $H(f_\Theta)$ denote the entropy of the random variable with density function $f_\Theta$, and $y_i$ denote an observation, constrained entropy maximization and maximum likelihood can be set up as follows:
\[ ML: \;\;\max_{f_\Theta \in \mathcal{F} }   \sum_i \log (f_\Theta(y_i) ) \; \;,\; \; ME: \;\; \max_{f_\Theta \in \mathcal{F} }   H(f_\Theta) \]
where $\mathcal{F} = \{ f_\Theta \;  | \; \bar{q}_{\alpha_i} (X_{f_\Theta}) \leq \hat{\bar{q}}_{\alpha_i}(X) \;\;\forall i=1,...,k \}$.

While we leave full exploration of this framework for future work, this simple formulation illustrates another potential use for the provided superquantile and bPOE formulas within traditional parametric frameworks.

\section{Conclusion}
In this paper, we first derived closed-form formulas for the superquantile and bPOE, then utilized them within parametric portfolio optimization and density estimation problems. We were able to derive superquantile formulas for a variety of distribution's, including ones with exponeitial tails (Exponential, Pareto/GPD, Laplace), symmetric distribution's (Normal, Laplace, Logistic, Student-t), and asymmetric distribution's with heavy tails (LogNormal, Weibull, LogLogistic, GEV). With bPOE formulas, while we had less success deriving truly closed-form solutions, we saw that it can still be calculated by solving a one-dimensional convex optimization problem or one-dimensional root finding problem.

We then utilized these formulas to develop two parametric procedures, one in portfolio optimization and one in density estimation. We first found that formulas for Normal, Laplace, Student-t, Logistic, and GEV are all distribution's which yield tractable superquantile and bPOE portfolio optimization problems. Furthermore, we found that bPOE-optimal portfolios are more robust to changing distributional assumptions compared to superquantile-optimal portfolios. Specifically, bPOE optimal portfolios are optimal, simultaneously, for an entire class of distributions. Finally, we presented a variation on the method of moments where moments are replaced by superquantile's. This parametric procedure is made possible by our closed-form formula's and we illustrate its use in the case of heavy tailed assymetric data, where additional emphasis on fitting the tail via superquantile conditions can be highly desirable. We find that this method makes it easy to direct the focus of the fitting procedure on tail samples.

\bibliographystyle{spbasic}
 \bibliography{Closed_Form_bPOE}

\end{document}